\documentclass[12pt, pdftex]{iopart}

\expandafter\let\csname equation*\endcsname\relax
\expandafter\let\csname endequation*\endcsname\relax

\usepackage[pdftex]{graphicx}
\usepackage{mathtools,amssymb,amsthm}
\usepackage{latexsym, bm,  braket, multirow, url}
\usepackage{ascmac}
\usepackage{arydshln}
\usepackage{paralist}
\usepackage{setspace}
\usepackage{wrapfig,here}
\usepackage{color}
\usepackage{cite}
\usepackage{cancel}
\usepackage{comment} 

\theoremstyle{plain}
\newtheorem{thm}{Theorem}
\newtheorem{prp}[thm]{Proposition}
\newtheorem{cor}[thm]{Corollary}
\newtheorem{lem}[thm]{Lemma}
\newtheorem{rem}[thm]{Remark}

\newtheorem{fct}[thm]{Fact}
\newtheorem{rmk}[thm]{Remark}

\newcommand{\ol}{\overline}
\newcommand{\ep}{\epsilon}
\newcommand{\ve}{\varepsilon}
\newcommand{\ceq}{\coloneqq}
\newcommand{\lto}{\longrightarrow}

\newcommand{\lsrj}{\relbar\joinrel\twoheadrightarrow}
\newcommand{\lrto}{\leftrightarrow}
\newcommand{\llrto}{\longleftrightarrow}
\newcommand{\midd}{\, | \,}
\newcommand{\dapp}{\, \stackrel{?}{\approx} \,}

\newcommand{\thf}{(1/2)}
\newcommand{\tmx}{\text{max}}
\newcommand{\tmix}{\text{mix}}
\newcommand{\bbC}{\mathbb{C}}
\newcommand{\bbE}{\mathbb{E}}

\newcommand{\bbN}{\mathbb{N}}

\newcommand{\bbR}{\mathbb{R}}
\newcommand{\bbV}{\mathbb{V}}

\newcommand{\bbZ}{\mathbb{Z}}

\newcommand{\frS}{\mathfrak{S}}

\newcommand{\clH}{\mathcal{H}}

\newcommand{\clU}{\mathcal{U}}
\newcommand{\clV}{\mathcal{V}}

\newcommand{\sfP}{\mathsf{P}}
\newcommand{\clN}{\mathfrak{N}}

\DeclareMathOperator{\av}{av}

\DeclareMathOperator{\id}{\mathsf{id}}

\DeclareMathOperator{\SU}{SU}
\DeclareMathOperator{\Prb}{Pr}

\DeclareMathOperator{\rank}{rank}

\newcommand{\kb}[1]{\left| #1 \right> \! \left< #1 \right|}
\newcommand{\ff}[2]{#1^{\downarrow #2}}
\newcommand{\hg}[2]{{}_{#1} F_{#2}}
\newcommand{\abs}[1]{\left| #1 \right|}

\newcommand{\tket}[1]{| #1 \rangle}

\newmuskip\HGmuskip
\newcommand{\HGcomma}{{\normalcomma}\mskip\HGmuskip}
\newcommand*\HG[6][8]{%
 \begingroup
  \HGmuskip=#1mu\relax
  \mathchardef\normalcomma=\mathcode`,
  \mathcode`\,=\string"8000
  \begingroup\lccode`\~=`\,
  \lowercase{\endgroup\let~}\HGcomma
  {}_{#2}F_{#3}{\left[\genfrac..{0pt}{}{#4}{#5};#6\right]}%
 \endgroup}
\newcommand*\qHG[7][8]{%
 \begingroup
  \HGmuskip=#1mu\relax
  \mathchardef\normalcomma=\mathcode`,
  \mathcode`\,=\string"8000
  \begingroup\lccode`\~=`\,
  \lowercase{\endgroup\let~}\HGcomma
  {}_{#2}\phi_{#3}{\left[\genfrac..{0pt}{}{#4}{#5};#6,#7\right]}%
 \endgroup}

\begin{document}

\title[Stochastic behavior of outcome of Schur-Weyl duality measurement]%
{Stochastic behavior of outcome of Schur-Weyl duality measurement}
\author{Masahito Hayashi$^{1,2,3}$, Akihito Hora$^{4}$, Shintarou Yanagida$^{3}$}
\address{$^1$ School of Data Science, The Chinese University of Hong Kong,
Shenzhen, Longgang District, Shenzhen, 518172, China}
\address{$^2$ International Quantum Academy (SIQA), Shenzhen 518048, China}
\address{$^3$ Graduate School of Mathematics, Nagoya University, 
 Furo-cho, Chikusa-ku, Nagoya, 464-8602, Japan}
\address{$^4$ Department of Mathematics, Faculty of Science, Hokkaido University, %
 Kita 10, Nishi 8, Kita-Ku, Sapporo, Hokkaido, 060-0810, Japan}
\eads{\mailto{hmasahito@cuhk.edu.cn},\mailto{hora@math.sci.hokudai.ac.jp}, 
 \mailto{yanagida@math.nagoya-u.ac.jp}}
\vspace{10pt}
\begin{indented}
\item[] November, 2023
\end{indented}

\begin{abstract}
We focus on the measurement defined by the decomposition based on
Schur-Weyl duality on $n$ qubits.
As the first setting, we discuss the asymptotic behavior 
of the measurement outcome when the state is given as
the permutation mixture 
$\rho_{mix,n,l}$
of the state $\ket{1^{l} \, 0^{n-l}}
:=\ket{1}^{\otimes l} \otimes \ket{0}^{\otimes (n-l)}$.
In contrast, when the state is given as the Dicke state
$ \ket{\Xi_{n,l}}$,
the measurement outcome takes one deterministic value.
These two cases have completely different behaviors. 
As the second setting,
we study the case when the state is given as the tensor product
of the permutation mixture $\rho_{mix,k,l}$
and the Dicke state $ \ket{\Xi_{n-k,m-l}}$.
We derive various types of asymptotic distribution including 
a kind of central limit theorem when $n$ goes to infinity.
\end{abstract}

\vspace{2pc}


\maketitle

\section{Introduction}\label{s:intro}
Schur-Weyl duality plays a central role in quantum information,
and it introduces direct sum decomposition on the $n$-tensor product system.
This decomposition makes a decomposition of the identity operator, and
can be regarded as a projection measurement over the corresponding quantum system,
so called Schur-Weyl duality measurement.
This measurement is used for the estimation for the eigenvalue of the density matrix,
universal quantum data compression, and
universal distortion-free entanglement concentration.
When we focus on the
$n$-tensor product of the qubit system,
each component of the decomposition is characterized by Young index $(n-x,x)$
with $x=0,1, \ldots, \lfloor n/2\rfloor$.
The outcome $X$ of the measurement 
is given by the second index $x$.
When the state is prepared as the $n$-tensor product of a state with 
eigenvalues $(1-p,p)$ under the condition $0<p< \frac{1}{2}$,
the variable $\frac{X}{n}$ converges to $p$
and the normalized difference 
$\frac{X-np}{\sqrt{n}}$ is asymptotically subject to the normal distribution 
with variance $p(1-p)$ \cite[(55)]{HM04}\cite{EBGMM}\cite[Section 6.3]{Ha}.
In contrast, when the state is prepared as the $n$-tensor product 
of the completely mixed state,
the variable $\frac{X}{n}$ converges to $1/2$
and the normalized difference 
$\frac{X-n/2}{\sqrt{n}}$ is asymptotically subject to 
the chi-square distribution \cite[Section 13.5]{FAA}.

For a further study, we remember that 
any $n$-tensor product state can be written as a mixture of 
equal weight states.
Once we fix our basis to the computation basis $\ket{0},\ket{1}$,
each equal weight state is given as 
the permutation mixture $\rho_{mix,n,l}$ 
of the state $\ket{1^{l} \, 0^{n-l}}\ceq
\ket{1}^{\otimes l} \otimes \ket{0}^{\otimes (n-l)}$.
Notice that 
the outcome $X$ of Schur-Weyl duality measurement 
has the same behavior with the state 
$\kb{1^{l} \, 0^{n-l}}$ and its permutation mixture $\rho_{mix,n,l}$.
The first question is the asymptotic behavior of 
the outcome $X$ when the state is
the permutation mixture $\rho_{mix,n,l}$
of a state $\kb{1^{l} \, 0^{n-l}}$.
When $l=np$ with $0<p<1/2$,
we show that 
the variable $\frac{X}{n}$ converges to $p$
and the normalized difference 
$\frac{X-np}{\sqrt{n}}$ converges to zero.
To see the behavior of the difference $\frac{X}{n}-p$, 
we need to enlarge it with larger scaling.
As the result, we show that $X-np$ 
is asymptotically subject to the geometric normal distribution.
In contrast, when $l=n/2$,
we show that
the variable $\frac{X}{n}$ converges to $1/2$,
but the normalized difference 
$\frac{X-n/2}{\sqrt{n}}$ does not converges to zero
and is asymptotically subject to the Rayleigh distribution.

On the other hand, when the state is the superposition 
$ \ket{\Xi_{n,l}}$
of the permutations of the state $\ket{1^{l} \, 0^{n-l}}$,
the outcome $X$ is always $0$
because such a state belongs to the symmetric subspace.
Such a state is called the Dicke state \cite{Di}, and 
it has been playing an important roles in 
calculation of entanglement measures \cite{WG,HMMOV,WEGM,W,ZCH}, 
quantum communication and quantum networking \cite{KSTSW,PCTPWKZ}. 
Some typical Dicke states have been realized in trapped atomic ions \cite{HCTW}. 
Recently, the multi-qubit Dicke state with half-excitations
has been employed to implement a scalable quantum search based on
Grover's algorithm by using adiabatic techniques \cite{IILV}. 

Now, we arise the following question.
How can we characterize the stochastic behavior of 
the distribution $P_{n,m,k,l}$ of the outcome $X$ of 
the Schur-Weyl duality measurement
when the state is given as the tensor product of 
$\rho_{mix,k,l}$ and $ \kb{\Xi_{n-k,m-l}}$?
When $k=l=0$, the state belongs to the symmetric subspace
and the variable 
$\frac{X}{n}$ takes the value $0$.
When $k=n, l=m$, the variable 
$\frac{X}{n}$ asymptotically converges to $m/n$ in probability.
Therefore, we expect that 
the variable $\frac{X}{n}$ takes values between $0$ and $m/n$.
Interestingly, as shown in the paper \cite{HHY1},
the probability mass function (pmf) $p(x \midd n,m,k,l)$ of 
the distribution $P_{n,m,k,l}$
can be expressed by using Hahn and Racah polynomials, which are 
$\hg{3}{2}$- and $\hg{4}{3}$-hypergeometric orthogonal polynomials.
Since the numerical calculation method of these functions
has been well established,
this method yields the explicit calculations of the distribution $P_{n,m,k,l}$.
Further, this method enables us to derive a useful recursion formula for 
the probability mass function (pmf) $p(x \midd n,m,k,l)$ of 
the distribution $P_{n,m,k,l}$. 

However, although these formulas give us the exact values of 
the pmf $p(x \midd n,m,k,l)$,
they did not directly present the trend of 
the distribution $P_{n,m,k,l}$ when $n$ is sufficiently large.
To extract their trend, we discuss 
their asymptotic behavior under two types of limit, i.e., Type I and Type II limits.
Type I limit assumes that $k$ and $l$ are fixed and $m$ is linear for $n$, 
i.e., $m=\xi n$ with fixed ratio $\xi$, under $n \to \infty$.
That is, the mixed state $\rho_{mix,k,l}$ is fixed and the size of the Dicke state $ \ket{\Xi_{n-k,m-l}}$ increases.
Type II limit does that $k,l$ and $m$ are linear for $n$ under $n \to \infty$.
That is, the sizes of the mixed state $\rho_{mix,k,l}$ and the Dicke state $ \ket{\Xi_{n-k,m-l}}$ increase.
Type I limit corresponds to the situation of the law of small numbers,
and Type II limit corresponds to the situation of the central limit theorem.
Under Type I limit, we show that 
the distribution $P_{n,m,k,l}$ converges to 
the convolution of two binomial distributions.

In Type II limit, i.e., the limit $n \to \infty$ with $k,l$ and $m$ linear for $n$,
we need to treat many ratios and parameters,
which are summarized in Table \ref{tab:intro:IIprm}.
While 
the fixed ratios 1 and 2 are natural,
the fixed ratio 3 is useful for our purpose.

\begin{table}[htbp]
\centering
\begin{tabular}{l|ll}
 fixed ratio set 1 
 & $\alpha=\frac{l  }{n}$ & $\beta=\frac{m-l}{n}$ \qquad
   $\gamma=\frac{k-l}{n}$ \qquad $\delta=\frac{n-m-k+l}{n}$ \\[1ex] 
 fixed ratio set 2 
 & $\xi=\frac{m}{n}$ &  $\kappa=\frac{k  }{n}$\phantom{-l} \qquad 
   $\tau=\frac{l  }{n}-\xi\kappa$ \\
[1ex] 
 fixed ratio set 3 
 & $\beta=\frac{m-l}{n}$ &  $\delta=\frac{n-m-k+l}{n}$\phantom{-l} \qquad 
   $\xi=\frac{m}{n}$  \\
   [1ex] \hline 
 limit pdf parameters 
 &$\mu \ceq \frac{1-\sqrt{D}}{2}$ &
  $\sigma^2  \ceq \frac{(1-\beta-\delta)\beta\delta}{D} $ 
  $D \ceq 4 \beta \delta+ (2\xi-1)^2$
 \end{tabular}
\caption{Ratios and parameters for Type II limit}
\label{tab:intro:IIprm}
\end{table}

Under this limit, we show that 
$\frac{X}{n}$ asymptotically converges to $\mu$ in probability,
where $\mu$ is defined in Table \ref{tab:intro:IIprm}.
We also derive a kind of central limit theorem
by using the above mentioned recursion formula.
That is, the normalized difference 
$\frac{X-np}{\sqrt{n}}$ is asymptotically subject to the normal distribution 
with the variance $\sigma$ defined in Table \ref{tab:intro:IIprm}
except for the case
$\beta \delta=0$ nor $\alpha =\gamma=0$.
That is,
denoting the probability for the distribution $P_{n,m,k,l}$
by $\Prb_n$, 
we have 
\begin{align}
 \lim_{n \to \infty} 
 \Prb_n\Bigl[t \le \frac{X-n \mu}{\sqrt{n} \sigma} \le u\Bigr]
 = \frac{1}{\sqrt{2\pi}} \int_t^u e^{-s^2/2} \, d s
\end{align}
for any real $t<u$.
Figure \ref{fig:intro:pi/Psi} shows some examples comparing the pmf $p(x \midd n,m,k,l)$
 and the normal distribution. 
\begin{figure}[htbp]
\centering
\includegraphics[width=.96\linewidth]{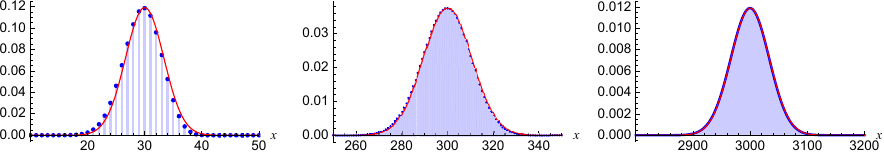}
\caption{Plots of $P_{n,m,k,l}[X=x]$ by blue dots and %
 the limit normal distribution by red lines %
 with $(m/n,k/n,l/n)=(0.4,0.6,0.3)$ fixed and %
 $n=100$ (left), $1000$ (middle), $10000$ (right).}
\label{fig:intro:pi/Psi}
\end{figure}

As shown in this paper,
the variance $\sigma$ is well defined and has non-zero value
except for these two cases.
When $\beta \delta=0$, 
the state $ \kb{\Xi_{n-k,m-l}}$ is 
$\rho_{n-k,0}$ or $\rho_{n-k,n-k}$.
In this case the outcome $X$ of Schur-Weyl duality measurement 
has the same behavior with the state 
$\kb{1^{l} \, 0^{n-l}}$ 
or $\kb{1^{l+n-k} \, 0^{k-l}}$. 
This case has been discussed in the above discussion.
When $\alpha =\gamma=0$, 
the state of the system to be measured is 
the state $ \kb{\Xi_{n-k,m-l}}$ so that 
the variable $\frac{X}{n}$ takes the value $0$.
Therefore, it is sufficient to discuss the case except for these two cases.
The above type of central limit theorem is the main contribution of this paper.
This analysis shows the intermediate case 
between the two cases mentioned the above.

Although several related preceding studies exist,
our obtained central limit theorem is different from theirs in the following way.
Several studies \cite{V,TV,Bi1,Bi2} discussed  
the asymptotic behavior of the distribution related to representation theory,
although their distributions do not depend on the state of the quantum system
while our distribution depends on the state of the quantum system.
In the context of quantum information, various studies 
\cite{Petz,Accardi1,Accardi2,guta-qubit,hayashi2008asymptotic, %
      guta-lan,qlan,hayashi2009quantum} 
discussed the central limit theorem related to unitary representation
in the sense of quantum Gaussian states dependently on the state of the quantum system.
However, since their analysis assumes the independent and identical tensor product structure,
these existing results cannot be applied to our setting
because our setting does not have the independent and identical tensor product structure.

In fact, Wigner 3j and 6j symbols are described by 
$\hg{3}{2}$- and $\hg{4}{3}$-hypergeometric orthogonal polynomials \cite{AW,VMK}.
However, their description and our description are different in the following points.
First, $\hg{3}{2}$- and $\hg{4}{3}$-hypergeometric orthogonal polynomials
have various parameters, and the parameter ranges of 
Wigner 3j and 6j symbols are different from our case.
Second, these polynomials are directly linked to the pmf in our case.
But, in the case of Wigner 3j and 6j symbols,
these polynomials express the coefficients in the superpositions,
and the squares of these polynomials express a certain distribution.
Therefore, it is not easy to derive a recursion formula for pmfs related to 
Wigner 3j and 6j symbols in a way similar to our case.

The remaining part of this paper is organized as follows.
Section \ref{S2-1} discusses 
the first setting, where the state is given as
$\rho_{mix,n,l}$.
Section \ref{S2-2-1} formulates
the second setting, where the state is given as
the tensor product of $\rho_{mix,k,l}$ and $ \kb{\Xi_{n-k,m-l}}$.
Section \ref{ss:intro:fin} reviews the results by \cite{HHY1}.
Section \ref{S2-2} states our obtained asymptotic results
in the second setting.
Section \ref{ss:I} proves our results under Type I limit.
Section \ref{sss:II:gen} presents 
our ad hoc derivation of the central limit theorem under Type II limit.
Section \ref{ss:AE-prf} proves the law of large numbers under Type II limit.
Section \ref{ss:CLT-prf} proves the central limit theorem under Type II limit.
Section \ref{ss:II:Cas} discusses the asymptotic behaviors of the expectation
and the variance under Type II limit.
Section \ref{ss:II:sp} discusses the tail probability under a special condition.
Section \ref{S-app} discusses two kinds of applications of our results.
Section \ref{S5} gives the conclusion and discussion.

\section{First setting}\label{S2-1}
This paper focuses the $n$-fold tensor product 
$\clH=(\bbC^2)^{\otimes n}$ of the qubit system $\bbC^2=\bbC \ket{0} \oplus \bbC \ket{1}$,
and consider it as the representation space of the permutation group $G=\frS_n$ 
acting by permutation of tensor factors as 
\begin{align} 
\pi(g)(\ket{i_1 \dotsm i_n}) \ceq \ket{i_{g^{-1}(1)} \dotsm i_{g^{-1}(n)}}, 
 \quad g \in \frS_n.
\end{align}
Also, we consider the representation of $\SU(2)$ as the product representation.

Our main focus is the $\SU(2)$-$\frS_n$ Schur-Weyl duality
on the tensor product system $\clH=(\bbC^2)^{\otimes n}$ depicted as
\begin{align}\label{eq:intro:SCS}
 \SU(2) \curvearrowright (\bbC^2)^{\otimes n} \curvearrowleft \frS_n.
\end{align}
The Schur-Weyl duality claims that we have the following decomposition of 
$(\bbC^2)^{\otimes n}$ into irreducible representations (irreps for short):
\begin{align}\label{eq:intro:SW}
 (\bbC^2)^{\otimes n} 
 = \bigoplus_{x=0}^{\lfloor n/2 \rfloor} \clU_{(n-x,x)} \boxtimes \clV_{(n-x, x)}.
\end{align} 
Here $\boxtimes$ denotes the tensor product of linear spaces, equipped with 
$\SU(2)$-action on the left and $\frS_n$-action on the right factor,
$\clV_{(n-x, x)}$ denotes the $\frS_n$-irrep corresponding to the partition $(n-x,x)$, 
and $\clU_{(n-x,x)}$ denotes the highest weight $\SU(2)$-irrep of dimension $n-2x+1$.
Using the hook formula (see \cite[I.5, Example 2; I.7, (7.6)]{M} for example),
we can compute the dimension of $\clV_{(n-x, x)}$ as 
\begin{align}\label{eq:hook}
 \dim \clV_{(n-x,x)} = \frac{k!}{\prod_{\square \in (n-x,x)} h(\square)} = 
 \frac{n-2x+1}{n-x+1} \binom{n}{x} = \binom{n}{x} - \binom{n}{x-1},
\end{align}
where $\square \in (n-x,x)$ denotes a box in the Young diagram 
corresponding to the partition $(n-x,x)$.
Using this formula, we have the completely mixed state 
$\rho_{\tmix, \clV_{(n-x,x)}} = \frac{1}{\dim \clV_{(n-x,x)}} \id_{\clV_{(n-x,x)}}$.
The highest weight $\SU(2)$-irrep of dimension $2j+1$ has the standard basis 
$ \{ \ket{j,m} \mid m=-j,-j+1, \dotsc, j\}$, where
$j$ and $m$ express the total and z-direction angular momenta, respectively.

In our analysis,  
the projector $\sfP_{x}$ onto the isotypical component 
in the decomposition \eqref{eq:intro:SW}
plays a key role, and is defined as
\begin{align}\label{eq:intro:proj}
 \sfP_{x}\colon (\bbC^2)^{\otimes n} \lsrj \clU_{(n-x,x)} \boxtimes \clV_{(n-x,x)},
\end{align}
where $x$ takes a value in $\{0,1, \ldots, \lfloor n/2\rfloor\}$.
As the set of the projectors 
$\{\sfP_{x}\}_{x}$ forms a projective measurement,
and is called
Schur-Weyl duality measurement.
The aim of this paper is to investigate the behavior of 
the measurement outcome $X$.

First, we assume that the state is given as  
the state $\kb{1^{l} \, 0^{n-l}}$ or
its permutation mixture $\rho_{mix,n,l}$ with $l \le n/2$, 
and they are defined as
\begin{align}\label{eq:intro:1^a}
   \ket{1^l \, 0^{n-l}} 
&\ceq \tket{\overbrace{1 \dotsm 1}^{l} \, \overbrace{0 \dotsm 0}^{k-l}}
 =  \ket{1}^{\otimes l} \otimes \ket{0}^{\otimes (n-l)} \in (\bbC^2)^{\otimes n}, \\
   \rho_{mix,n,l} &\ceq \tbinom{n}{l}^{-1}
 \bigl(\kb{1^l \, 0^{n-l}} + \text{permuted terms}\bigr).
\end{align}
Since the projection $ \sfP_{x}$ is permutation invariant,
we have
\begin{align}\label{NMU}
 p(x|n,l) \ceq \tr \sfP_{x} \kb{1^{l} \, 0^{n-l}}=\tr \sfP_{x} \rho_{mix,n,l}.
\end{align}
This probability can be calculated as
\begin{align}\label{MNER}
 p(x|n,l)
 = \frac{\binom{n}{x} - \binom{n}{x-1}}{\binom{n}{l}}
= \frac{\binom{n}{x}}{ \binom{n}{l}} \frac{n-2x+1}{n-x+1}
 \end{align}
for $x=0,1,\dotsc,l$ 
as follows.
The projection $\binom{n}{l}\rho_{mix,n,l}$ commutes with $\sfP_{x}$ and
satisfies the relation
\begin{align}\label{NJNGT}
\rank
 \Big(\sum_{x=0}^{x'} \sfP_{x}\Big)
\binom{n}{l}\rho_{mix,n,l} \Big(\sum_{x=0}^{x'} \sfP_{x}\Big)
= \binom{n}{x'}
 \end{align}
for $ x'=0,1,\dotsc,l $. 
Then, \eqref{NJNGT} implies \eqref{MNER}.

Next, we discuss the asymptotic behavior when 
the ratio $\frac{l}{n}$ is fixed to $\mu$.
Let us denote by $\Prb_n$ the probability for the distribution 
defined in \eqref{MNER}.

\begin{thm}\label{thm:II:bd=0:A}
When $\mu <\frac{1}{2}$,
we have
\begin{align}\label{eq:II:bd=0:BA1}
 \lim_{n \to \infty} \Prb_n\bigl[X=n\mu-i \bigr]
=\Bigl(\frac{\mu}{1-\mu}\Bigr)^i \, \frac{1-2\mu}{1-\mu}.
\end{align}
That is, the variable $n \mu -X$ asymptotically obeys 
the geometric distribution with parameter $\frac{\mu}{1-\mu}$.
We also have
\begin{align}\label{eq:II:bd=0:AU12}
 \bbE[X] = n \mu- \frac{\mu}{1-2\mu} + o(n^0), \quad 
 \bbV[X] = \frac{\mu(1-\mu)}{(1-2\mu)^2} + o(n^0). 
\end{align}
\end{thm}

\begin{proof}
First, we consider the case with $\mu<\frac{1}{2}$, i.e., $l<n-l$.
The probability \eqref{MNER} is rewritten as
\begin{align}\label{eq:bd=0:p2}
 p(l-i) = \frac{l(l-1)\cdots (l-i+1)}{(n-l+1)(n-l+2)\cdots (n-l+i)}
 \frac{n-2 l+2i+1}{n-l+i+1}.
\end{align}
Taking the limit $n \to \infty$, we obtain \eqref{eq:II:bd=0:BA1}.
{}To prove \eqref{eq:II:bd=0:AU12}, it is enough to show
\begin{align}\label{eq:bd=0:AU34}
 \lim_{n \to \infty} \bbE[ X-n \mu+\tfrac{\mu}{1-2\mu}  ] = 0, \quad 
 \lim_{n \to \infty} \bbE\bigl[\bigl(X-n \mu+\tfrac{\mu}{1-2\mu}\bigl)^2 \bigr]
 = \frac{\mu(1-\mu)}{(1-2\mu)^2}. 
\end{align}
Denoting $\phi \ceq \frac{\mu}{1-\mu}$, we have $0 \le \phi<1$, and can rewrite 
\eqref{eq:II:bd=0:BA1} as $\lim_{n \to \infty} p(n\mu-i)=\phi^i(1-\phi)$.
Then, we can deduce the following relations:
\begin{align}\label{eq:bd=0:AU56}
\begin{split}
  \lim_{R \to \infty} \lim_{n \to \infty} 
  \sum_{i=0}^R \bigl(i-\tfrac{\mu}{1-2\mu}\bigr) p(n\mu-i) 
&=\lim_{R \to \infty}-(R+1) \phi^{R+1} = 0,  \\
  \lim_{R \to \infty}\lim_{n \to \infty} 
  \sum_{i=0}^R \bigl(i-\tfrac{\mu}{1-2\mu}\bigr)^2 p(n\mu-i) 
&=\lim_{R \to \infty} \Bigl(\tfrac{\phi(1-\phi^{R+1})}{(1-\phi)^2}
  -(R+1)^2 \phi^{R+1}\Bigr) \\
&= \tfrac{\phi}{(1-\phi)^2} = \tfrac{\mu(1-\mu)}{(1-2\mu)^2}.
\end{split}
\end{align}
On the other hand, using \eqref{eq:bd=0:p2}, we have $p(l-i)\le \phi^i$. Hence, 
\begin{align} \label{eq:bd=0:AU78}
\begin{split}
 \lim_{R \to \infty} \lim_{n \to \infty} 
 \sum_{i=R+1}^{n} \abs{i- \tfrac{\mu}{1-2\mu}} p(n\mu-i) & \le
 \lim_{R \to \infty} \sum_{i=R+1}^{\infty} 
 \abs{i- \tfrac{\phi}{1-\phi}} \phi^i = 0, \\
 \lim_{R \to \infty} \lim_{n \to \infty} 
 \sum_{i=R+1}^{n} \bigl(i- \tfrac{\mu}{1-2\mu}\bigr)^2 p(n \mu-i) &=
 \lim_{R \to \infty} \sum_{i=R+1}^{\infty} 
 \bigl(i- \tfrac{\mu}{1-2\mu}\bigr)^2 \bigl(\tfrac{\mu}{1-\mu}\bigr)^i = 0. 
\end{split}
\end{align}
Therefore, combining \eqref{eq:bd=0:AU56} and \eqref{eq:bd=0:AU78}, 
we obtain \eqref{eq:bd=0:AU34}.
\end{proof}

\begin{thm}\label{thm:II:bd=0:B}
When $\mu=\frac{1}{2}$,
we have 
\begin{align}\label{eq:II:bd=0:AN1}
 \lim_{n \to \infty} \Prb_n\Bigl[-u \le \frac{X-n \mu}{\sqrt{n}} \le -t\Bigr]
 = \int_t^u 4s e^{-2 s^2} \, d s
\end{align}
for any real $t<u$.
That is, $\frac{X-n \mu}{\sqrt{n}}$ asymptotically obeys the Rayleigh distribution.
We also have 
\begin{align}\label{eq:II:bd=0:NH12}
 \bbE[X] = n \mu- \sqrt{n}\sqrt{\frac{\pi}{8}} + o(\sqrt{n}), \quad 
 \bbV[X] = n\Bigl(\frac{1}{2}-\frac{\pi}{8}\Bigr) + o(n).
\end{align}
\end{thm}

\begin{proof}
Similarly as \eqref{eq:bd=0:p2}, we can rewrite the probability \eqref{MNER} as
\begin{align}\label{eq:bd=0:NK1}
 p(n\mu - \sqrt{n}R) = \frac{2 \sqrt{n}R+1}{\frac{n}{2}+\sqrt{n}R+1}
 \prod_{i=1}^{n \sqrt{R}} \Bigl(\frac{1-\frac{2 (i-1)}{n}}{1+\frac{2 i}{n}} \Bigr).
\end{align}
Hence, using $\log \frac{1-x/n}{1+y/n} = x+y+y^2+o(n^{-2})$, we have 
\begin{align}
&\lim_{n \to \infty}\log \bigl( \sqrt{n} p(n\mu-\sqrt{n}R) \bigr)
 = (\log 4 R) \lim_{n \to \infty} \sum_{i=1}^{R \sqrt{n}} 
    \log \frac{1-\frac{2 (i-1)}{n}}{1+\frac{2 i}{n}} \\
&= (\log 4 R) \lim_{n \to \infty} \frac{1}{\sqrt{n}} \sum_{i=1}^{R \sqrt{n}}
    \Bigl(-\frac{2 (i-1)}{\sqrt{n}}-\frac{2 i}{\sqrt{n}} \Bigr) 
 = (\log 4 R) \int_{0}^R (-4s) \, d s = - 2 R^2 (\log 4 R).
\end{align}
Therefore, we have
\begin{align}
  \lim_{n \to \infty} \Prb_n \bigl[ -u \le \tfrac{X-n \mu}{\sqrt{n}} \le -t \bigr]
&=\lim_{n \to \infty} \sum_{i= t \sqrt{n}}^{ u \sqrt{n}} p(n \mu -i) 
 =\lim_{n \to \infty} \sum_{i= t \sqrt{n}}^{ u \sqrt{n}} 
  \frac{1}{\sqrt{n}} \cdot \frac{4i}{\sqrt{n}} e^{-2 (i/\sqrt{n})^2} 
 \nonumber \\
&=\int_t^u 4s e^{-2 s^2} \, d s,
  \label{eq:bd=0:NA1}
\end{align}
which implies \eqref{eq:II:bd=0:AN1}.
The distribution defined by \eqref{eq:bd=0:NA1} is the Rayleigh distribution
with mean $\sqrt{\frac{\pi}{8}}$ and variance $\frac{1}{2}-\frac{\pi}{8}$.
Then, due to the same reason as the proof of \eqref{eq:II:bd=0:AU12} in 
Theorem \ref{thm:II:bd=0:A}, it is sufficient 
for the proof of \eqref{eq:II:bd=0:NH12} to show the following:
\begin{align}\label{eq:II:bd=0:AH6}
 \lim_{R \to \infty}\lim_{n \to \infty} \sum_{i=\sqrt{n}R}^{n} 
 \abs{\tfrac{i- \sqrt{n} \sqrt{\frac{\pi}{8}}}{\sqrt{n}}} \, p(n \mu +i) 
 =0, \quad 
 \lim_{R \to \infty}\lim_{n \to \infty} 
 \sum_{i=\sqrt{n}R}^{n} 
 \Big( \tfrac{i- \sqrt{n}\sqrt{\frac{\pi}{8}}}{\sqrt{n}}\Big)^2
 p(n \mu +i) =0.
\end{align}
Since the second statement implies the first, we will show only the second one.

Since $x-x^2\le \log (1+x) \le x$ for $x>0$,
we have $\log \frac{1-x}{1+y} \le x+y+y^2$ for $x,y>0$.
Hence, \eqref{eq:bd=0:NK1} implies
\begin{align}
\begin{split}
 \frac{\frac{n}{2}+\sqrt{n}R+1}{2 \sqrt{n}R+1} \log p(n \mu -j) 
 \le &\sum_{i=1}^j \Bigl(-\frac{2(i-1)}{n}-\frac{2i}{n}+\Bigl(\frac{2i}{n}\Bigr)^2\Bigr) \\
&= -\frac{2j^2}{n}+\frac{2j (j+1)(2j+1)}{ 3n^2}.
\end{split}
\end{align}
Then, as for the second statement in \eqref{eq:II:bd=0:AH6}, we have 
\begin{align}
 \sum_{i=\sqrt{n}R}^n
 \Bigl( \frac{i-\sqrt{n}\sqrt{\frac{\pi}{8}}}{\sqrt{n}}\Bigr)^2 p(n \mu +i)
 \le \sum_{i=\sqrt{n}R}^{n} \frac{2 i+1}{\frac{n}{2}+i+1}
 \Big( \frac{i- \sqrt{n}\sqrt{\frac{\pi}{8}}}{\sqrt{n}}\Big)^2
 e^{-\frac{2j^2}{n}+\frac{2j (j+1)(2j+1)}{ 3n^2}}.
\end{align}
Hence, 
\begin{align}
 \lim_{n \to \infty} \sum_{i=\sqrt{n}R}^n 
 \Bigl(\frac{i-\sqrt{n}\sqrt{\frac{\pi}{8}}}{\sqrt{n}}\Bigr)^2 p(n \mu +i) \le 
 \int_{R}^\infty 4 s \Bigl( s- \sqrt{\frac{\pi}{8}}\Bigr)^2 e^{-2s^2} \, d s < \infty,
\end{align}
which implies the second statement in \eqref{eq:II:bd=0:AH6}.
Thus the proof is completed.
\end{proof}

\begin{rmk}
The asymptotic expectation in \eqref{eq:II:bd=0:AU12} and \eqref{eq:II:bd=0:NH12}
can also be deduced from the integral expression of $\bbE[X]$. 
By a straightforward calculation using \eqref{MNER}, we have 
\begin{align}
 \bbE[X] = 
     l - \frac{1}{\binom{n}{l}} \sum_{x=0}^{  l-1} \binom{n}{x}. 
\end{align}
We can further rewrite it by the formula of partial sum of 
binomial coefficients in terms of incomplete beta function \cite[p.52, (3-3)]{WB} as
\begin{align}\label{eq:MN=0:E}
 \bbE[X] = l - 2^n l \int_0^{1/2} t^{n-l} (1-t)^{l-1} \, d t 
  = n \mu - n \mu \int_0^1 \frac{s^{n(1-\mu)} (2-s)^{n \mu}}{2-s} d s.
\end{align}
The integration part
\begin{align}
 \int_0^1 \frac{s^{n(1-\mu)} (2-s)^{n \mu} }{2-s} d s
 = \int_0^1 \frac{1}{2-s} e^{n f_1(s)} \, d s, \quad 
 f_1(s) \ceq (1-\mu) \log s+ \mu \log(2-s)
\end{align}
can be estimated by Laplace's method, and we have 
\begin{align}
 \bbE[X] = n \mu - 
 \begin{cases} \frac{\mu}{1-2\mu} + o(n^{0}) & (\xi<1/2) \\ 
 \sqrt{n} \sqrt{\frac{\pi}{8}} + o(\sqrt{n}) & (\xi=1/2) 
 \end{cases},
\end{align}
which recovers the expectation part of 
\eqref{eq:II:bd=0:AU12} and \eqref{eq:II:bd=0:NH12}.
\end{rmk}

\section{Second setting}
\subsection{Formulation}\label{S2-2-1}
Next, we focus on a Dicke state, which is a typical invariant state because 
it is given as the permutation invariant state 
with fixed weights $N$ and $M$ as
\begin{align}\label{eq:0:Dicke}
 \ket{\Xi_{N+M,M}} \ceq \tbinom{N+M}{M}^{-1/2} 
 \bigl(\ket{1^M \, 0^N} + \text{permuted terms}\bigr)
 \in (\bbC^2)^{\otimes (N+M)}.
\end{align}
When the state is $ \ket{\Xi_{n,l}}$,
the outcome $X$ of Schur-Weyl duality measurement 
is always zero due to the symmetry.
Then, we discuss the behavior of the outcome $X$ when 
the state is given as the tensor product state
$\rho_{mix,k,l}\otimes \kb{\Xi_{n-k,m-l}}$.
For brevity, we denote the range of the parameters $n,m,k,l$ by
\begin{align}\label{eq:intro:clN}
 \clN \ceq \{(n,m,k,l) \in \bbZ_{\ge 0}^4 \mid
  m, k \le n, \ m+k-n \le l \le m \wedge k\}.
\end{align}
Our focus is the pmf defined as
\begin{align}\label{eq:intro:p(x)}
  p(x \midd n,m,k,l) 
\ceq& \tr \sfP_{x} \rho_{mix,k,l}\otimes \kb{\Xi_{n-k,m-l}} \nonumber \\
  = & \tr \sfP_{x} \kb{1^{l} \, 0^{k-l}}\otimes \kb{\Xi_{n-k,m-l}}.
\end{align}
By using the state
\begin{align} 
 \ket{\Xi_{n,m|k,l}} \ceq \ket{1^l \, 0^{k-l}} \otimes \ket{\Xi_{n-k,m-l}}
 \in (\bbC^2)^{\otimes n},
\end{align}
the pmf $ p(x \midd n,m,k,l) $ has the form
\begin{align} 
p(x \midd n,m,k,l)= \bra{\Xi_{n,m|k,l}} \sfP_{(n-x,x)} \ket{\Xi_{n,m|k,l}}.
\end{align}

The setting implies a few but basic properties of $p(x)$. 
First, we can switch $\ket{0} \lrto \ket{1}$ in the system, 
and $p(x)$ inherits the symmetry
\begin{align}\label{eq:intro:sym}
 p(x \midd n,m,k,l) = p(x \midd n,n-m,k,k-l).
\end{align}
Second, since $p(x)$ comes from the decomposition \eqref{eq:intro:SW} 
of unitary representation, $p(x)$ is a probability mass function
for a discrete probability distribution $\{0,1,\dotsc,\lfloor n/2\rfloor\} \to \bbR_{\ge 0}$.

When $m=l$, the Dicke state 
$ \ket{\Xi_{n,m|k,l}}$ is $ \ket{ 0^{n-k}}$. We have
\begin{align}\label{BG1}
  p(x \midd n,m,k,l) 
  = & \tr \sfP_{x} \kb{1^{l} \, 0^{k-l}}\otimes \kb{ 0^{n-k}}
=p(x \midd n,l) ,
\end{align}
which does not depend on $k$.
Similarly, when $n-m-k+l=0$, the Dicke state $ \ket{\Xi_{n,m|k,l}}
$ is $ \ket{ 1^{k-l}}$. We have
\begin{align}\label{BG2}
  p(x \midd n,m,k,l) 
  = & \tr \sfP_{x} \kb{1^{l} \, 0^{k-l}}\otimes \kb{ 1^{k-l}}
=p(x \midd n,k) ,
\end{align}
which does not depend on $m,l$.
Hence, these two cases are reduced to the first setting,
and has redundancy for indices.
Also, when $k=0$, our state is the Dicke state 
$\ket{\Xi_{n,m|k,l}}$ so that the random variable $X$ takes $0$ with probability $1$.  

\subsection{Explicit formulas of $P_{n,m,k,l}$}\label{ss:intro:fin}
This section reviews several formulas for the distribution $P_{n,m,k,l}$ obtained in \cite{HHY1}.

\begin{thm}
\label{thm:Hahn}
Assume $m \le n-m$ and $x \le m$. Then, we have the following two kinds of the representation of the pmf $p(x \midd n,m,k,l)$.
\begin{align}
p(x \midd n,m,k,l) 
=& \frac{\binom{n}{x}}{\binom{n}{m}} \, \frac{n-2x+1}{n-x+1} \, 
 \sum_{i=0}^{M \wedge N} \binom{M}{i} \binom{N}{i} \, \omega_{(n-x,x)}(i), \label{eq:main-formula}\\
= &\binom{n-k}{m-l} \frac{\binom{n}{x}}{\binom{n}{m}} \frac{n-2x+1}{n-x+1} 
 \HG{4}{3}{-x,x-n-1,-M,-N}{-m,m-n,-M-N}{1},\label{NMF}
\end{align}
where $M \ceq m-l$, $N \ceq n-m-k+l$, and 
\begin{align}\label{eq:sph-omega}
 \omega_{(n-x, x)}(i) = \binom{m}{i}^{-1} \binom{n-m}{i}^{-1}
 \sum_{r=0}^{i \wedge x} (-1)^r \binom{x}{r} \binom{m-x}{i-r} \binom{n-m-x}{i-r}.
\end{align}
We also used the standard notation of hypergeometric series \cite{GR}:
\[
 \HG{4}{3}{a_1,a_2,a_3,a_4}{b_1,b_2,b_3}{z} \ceq 
 \sum_{i=0}^\infty \frac{(a_1)_i(a_2)_i(a_3)_i(a_4)_i}{(1)_i(b_1)_i(b_2)_i(b_3)_i} z^i, 
 \quad 
 (a)_i \ceq a(a+1)\cdots(a+i-1).
\]
\end{thm}

Figure \ref{figure-composite} plots of the pmf $p(x \midd n,m,k,l)$ 
for a few examples of $(n,m,k,l)$. 

\begin{figure}[t]
	\centering
	\includegraphics[width=8cm]{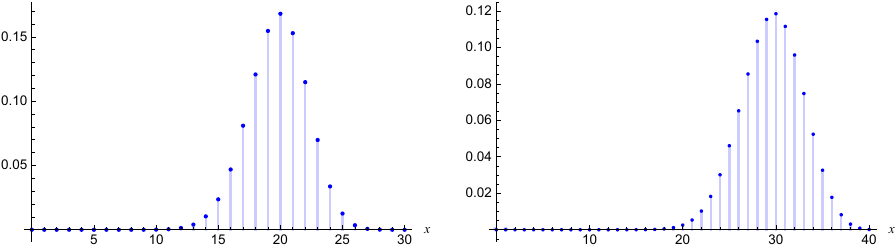}
\caption{Plots of $p(x \midd n,m,k,l)$ with $(n,m,k,l)=(100,30,40,20)$ %
 in left and $(n,m,k,l)=(100,40,60,30)$ in right.}
	\label{figure-composite}
\end{figure}

This distribution has simple forms in several special cases.
In the case $k=l$, we have the following proposition as the special case of \eqref{NMF}

\begin{prp}\label{prp:k=l}
Under the assumption $m \le n-m$ and $k=l$, we have 
\begin{align}
 p(x) = \frac{\binom{n}{x}}{\binom{n}{m}} \frac{\binom{l}{x}}{\binom{m}{x}}
        \frac{n-2x+1}{n-x+1} \binom{n-l-x}{m-l} 
\end{align} 
for $x \le l \wedge m$, and $p(x)=0$ otherwise.
\end{prp}

Further, based on the formula \eqref{NMF},
we have the following recursion formula,
which turns out to be a key lemma in the asymptotic analysis of $p(x)$.

\begin{lem}[{Three-term recurrence relation of $p(x)$}]
\label{lem:intro:rec}
In the case $m \le n-m$ and $0 \le x \le m$, we have
\begin{align}
\begin{split}
 a_x \frac{n-2x-1}{n-x} \binom{n}{x+1}^{-1} p(x+1) 
&-(a_x+c_x-M N) \frac{n-2x+1}{n-x+1} \binom{n}{x}^{-1} p(x) \\
&+c_x \frac{n-2x+3}{n-x+2} \binom{n}{x-1}^{-1} p(x-1) = 0
\end{split}
\end{align}
with $M \ceq m-l$ and $N \ceq n-m-k+l$.
Here we regard the third term as $0$ in the case $x=0$, 
and regard the first term as $0$ in the case $x=m$.
The coefficients $a_x$ and $c_x$ are given by 
\begin{align}
 a_x &\ceq \frac{(m-x)(n-m-x)(n-k-x)(n-x+1)}{(n-2x)(n-2x+1)}, \\
 c_x &\ceq \frac{x(x-k-1)(m-x+1)(n-m-x+1)}{(n-2x+1)(n-2x+2)}.
\end{align}
\end{lem}

In addition, we have the following relation between the variance and expectation.

\begin{thm}\label{thm:Cas}
Let $X$ be a random variable distributed by $p(x) = p(x \midd n,m,k,l)$.
Then, the expectation $\bbE[X]$ and the variance $\bbV[X]$ are related as 
\begin{align}
 \frac{\bbV[X]}{n^2} = 
 \frac{\bbE[X]}{n}\left(1-\frac{\bbE[X]}{n}\right)+\frac{\bbE[X]}{n^2} - \eta,
\end{align}
where we put
\begin{align}\label{eq:Cas:eta}
 \eta \ceq \frac{n(n-m)-(m-l)(n-m-k+l)}{n^2}.
\end{align}
\end{thm}

\section{Asymptotics analysis on distribution $P_{n,m,k,l}$}
\label{S2-2}
We study the asymptotics of our distribution $P_{n,m,k,l}$
under the limit $n \to \infty$. 
We study the following two cases.
\begin{enumerate}
\item
 The limit $n \to \infty$ with $k,l$ and $m/n$ fixed. 
 That is, the mixed state $\rho_{mix,k,l}$ is fixed, and the size of the Dicke state $ \ket{\Xi_{n-k,m-l}}$ increases.
\item
 The limit $n \to \infty$ with $m/n$, $k/n$ and $l/n$ fixed. 
That is, the sizes of the mixed state $\rho_{mix,k,l}$ and the Dicke state $ \ket{\Xi_{n-k,m-l}}$ increase.
\end{enumerate}
We call them Type I and Type II limit, respectively.
The details are given in Subsection \ref{sss:intro:I} and 
Subsections \ref{sss:intro:II}. 
Hereafter we denote by $\Prb_n$ the probability for the distribution 
$P_{n,n \xi,n \kappa,n \alpha}$, and by $X$ a random variable distributed
subject to the distribution.

\subsection{Asymptotic analysis of Type I limit}\label{sss:intro:I}
To address Type I limit, we employ the binomial distribution $B_{\xi,j}$ with $j$ trials
with successful probability $0 \le \xi \le 1$ 
whose pmf is given by $x \mapsto \binom{j}{x} \xi^x (1-\xi)^{j-x}$.
Then, we introduce the distribution
$ P_{R|\xi;k,l} \ceq B_{\xi,k-l} * B_{1-\xi,l}$
where
$*$ denotes the convolution of probability distributions. 
Its pmf is denoted by $q(x \midd \xi;k,l)$
and written as 
\begin{align}\label{eq:I:q2}
 q(x \midd \xi;k,l) = \xi^{l-x} (1-\xi)^{k-l-x} \sum_{u=0 \vee (x-k+l)}^{x \wedge l} 
 \binom{k-l}{x-u} \binom{l}{u} \xi^{2(x-u)} (1-\xi)^{2u}.
\end{align}
Then, as shown in Section \ref{thm:intro:I}, we have the following theorem.

\begin{thm}
\label{thm:intro:I}
In the limit $n \to \infty$ with $k,l$ and $\xi=m/n$ fixed, 
the distribution $P_{n,m,k,l}$ is approximated to 
the above defined distribution
$ P_{R|\xi;k,l}$ 
up to $O(1/n)$.
In particular, the asymptotic expectation is given by
\begin{align}\label{eq:intro:Eq}
 \lim_{n \to \infty} \bbE[X] = (k-l)\xi+l(1-\xi).
\end{align}
\end{thm}

In fact, it is known as the law of small numbers 
 that the distribution $B_{c/n,n}$ converges to a Poisson distribution
 as $n$ goes to infinity with a fixed number $c$.
In Type I limit, $k$ and $l$ are fixed, and correspond to 
the fixed number $c$ of the law of small numbers. 
That is, we can consider In Type I limit as 
the law of small numbers in our setting.

\subsection{Asymptotic analysis of type II limit}\label{sss:intro:II}
We discuss Type II limit.
Using the parametrization $\xi=m/n$ in the previous Theorem \ref{thm:intro:I}, 
we can label the fixed ratios as
\begin{align}\label{eq:intro:xkl}
    \xi = \frac{m}{n}, \quad 
 \kappa = \frac{k}{n}, \quad 
 \tau = \frac{l}{n}-  \xi \kappa.
\end{align}
Let us also introduce another parametrization of the fixed ratios:
\begin{align}\label{eq:intro:abcd}
 \alpha = \frac{l}{n},   \quad 
  \beta = \frac{m-l}{n}, \quad 
 \gamma = \frac{k-l}{n}, \quad
 \delta = \frac{n-m-k+l}{n}.
\end{align}
For example, the condition $(n,m,k,l) \in \clN$ in \eqref{eq:intro:clN} is equivalent to 
\begin{align}\label{eq:intro:II:pasmp}
 \alpha, \beta, \gamma,\delta \ge 0. 
\end{align}
These parameters have the following meaning.
$ \xi$ shows the ratio of $\ket{1}$.
$ \kappa$ shows the ratio of qubits in the mixed state $\rho_{mix,k,l}$.
To grasp their meaning of other parameters, we virtually 
introduce two binary random variable $Y',Z'$
such that 
\begin{align}
P(Y'=1,Z'=1)&=\alpha,~
P(Y'=1,Z'=0)=\beta, \\
P(Y'=0,Z'=1)&=\gamma,~
P(Y'=0,Z'=0)=\delta.
\end{align}
In this case, their marginal distributions are given as 
\begin{align}
P(Y'=1)=\xi,~ P(Z'=1)=\kappa.
\end{align}
$\tau$ expresses the correlation $Y'$ and $Z'$ as
\begin{align}
\tau= P(Y'=1,Z'=1)- P(Y'=1)P(Z'=1).
\end{align}
When $Y'$ and $Z'$ are independent, $\tau=0$.
The cases $\beta=0$ and $\delta=0$
correspond to \eqref{BG1} and \eqref{BG2}, and is reduced to the first setting.
The case $\kappa=0$ corresponds to the case when the random variable $X$
takes $0$ with probability $1$.
The parameters in \eqref{eq:intro:abcd} can be written by using the parameters
in \eqref{eq:intro:xkl} as
\begin{align}
\alpha=\xi \kappa +\tau ,~
\beta=\xi (1-\kappa) -\tau, ~
\gamma=(1-\xi) \kappa -\tau ,~
\delta=(1-\xi)(1- \kappa) +\tau .
\end{align}

However, in the following, we employ the set of three parameters $\beta,\delta,\xi \ge 0$, which are also free parameters
under the conditions $ \xi \ge \beta $ and $1-\xi \ge \delta$.
We introduce the fundamental quantities
\begin{align}\label{eq:intro:msD}
      D &\ceq 
4 \beta \delta+ (2\xi-1)^2
=      1-4(\alpha\gamma+\alpha\delta+\beta\gamma)
\\
&=4 (  (1-\xi)\xi(1- \kappa)^2-\tau (1-2\xi)(1- \kappa)-\tau^2)
+ (2\xi-1)^2, \\
    \mu &\ceq \frac{1-\sqrt{D}}{2}.
\end{align}

Then, we have the law of large numbers as follows.

\begin{thm}
\label{thm:intro:E}
Consider the limit $n \to \infty$ with fixed ratios $\alpha=\frac{l}{n}$, 
$\beta =\frac{m-l}{n}$, $\gamma=\frac{k-l}{n}$ and $\delta=\frac{n-m-k+l}{n}$.
Then, for any $\ve \in \bbR_{>0}$, we have 
\begin{align}
 \lim_{n \to \infty} \Prb_n\Bigl[\abs{\tfrac{X}{n}-\mu}>\ve\Bigr]=0.
\end{align}
In particular, the expectation $\bbE[X]$ behaves as 
\begin{align}
 \bbE[X] = n \mu + o(n).
\end{align}
\end{thm}

Figure \ref{fig:intro:cdf} displays the law of large numbers.

\begin{figure}[htbp]
\centering
\includegraphics[width=.96\linewidth]{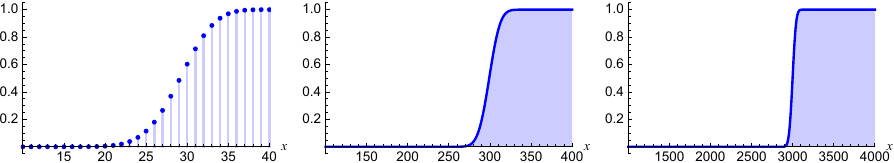}
\caption{Plots of the cdf $P_{n,m,k,l}[X \le x]$ with $(m/n,k/n,l/n)=(0.4,0.6,0.3)$ %
 fixed and $n=100$ (left), $1000$ (middle), $10000$ (right).}
\label{fig:intro:cdf}
\end{figure}


Figure \ref{fig:intro:Edivn} displays the behavior of $\bbE[X]/n$
when $\tau=0$ and $\xi$ is fixed.

\begin{figure}[htbp]
\centering
\includegraphics[width=.96\linewidth]{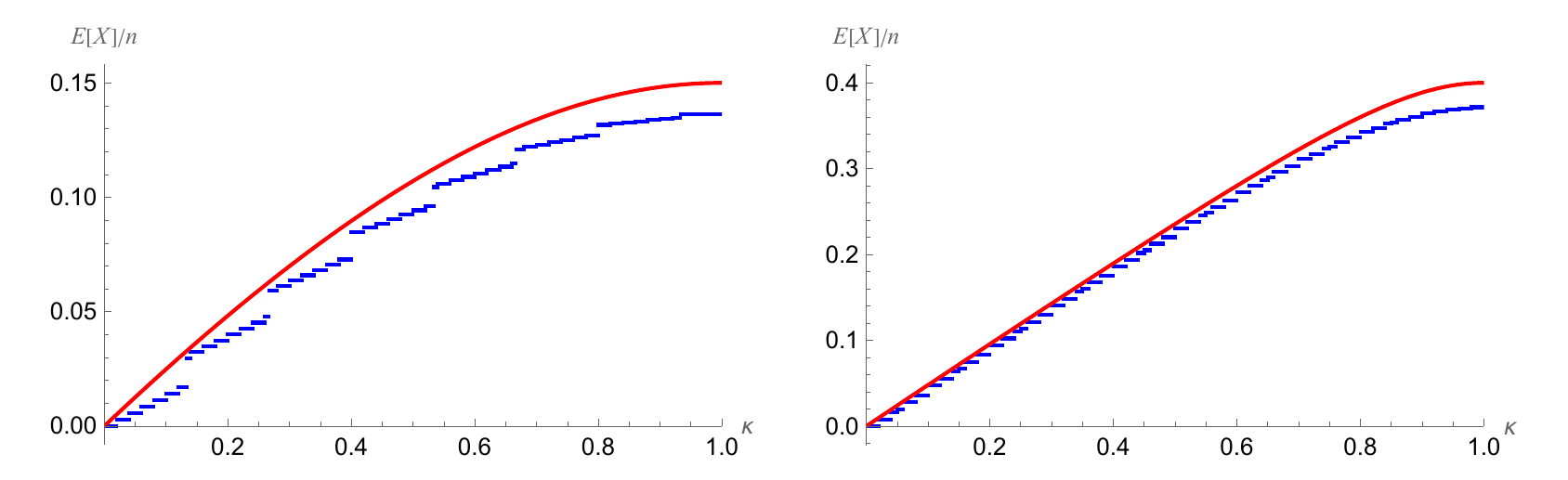}
\caption{Plots of $\bbE[X]/n$ as functions of $\kappa$ with $n=50$, $\tau=0$, 
and $\xi=0.15$ (left), $\xi=0.40$ (right).
The red lines show the $n \to \infty$ limit expectation $\mu$.}
\label{fig:intro:Edivn}
\end{figure}

For further analysis, we introduce the variance $\sigma\ge 0$ as
\begin{align}\label{eq:intro:msD2}
 \sigma^2 \ceq 
 \frac{(1-\beta -\delta)\beta \delta}{4 \beta \delta+ (2\xi-1)^2}
={\frac{(\alpha+\gamma)\beta\delta}{D}}
\end{align}
unless
$\xi=\frac{1}{2}$ and
$ \beta \delta=0$.
Also, the variance $\sigma^2$ is zero
if an only if
$\beta +\delta=1$, i.e., $\alpha=\gamma=0$, or
$\xi\neq \frac{1}{2}$ and $ \beta \delta=0$.
Hence, when the relation $\alpha=\gamma=0$ nor $\beta \delta=0$
does not hold, i.e., 
\begin{align}\label{eq:intro:II:asmp}
 \alpha+\gamma=1-\beta- \delta, \beta, \delta > 0,
\end{align}
the variance $\sigma^2$ is well defined and is strictly positive.

In fact, $ \sigma^2$ has a pathological behavior in the neighborhood
$\xi=\frac{1}{2}$ and
$ \beta \delta=0$.
To see this behavior, we consider $ \sigma^{-2}$ as
\begin{align}\label{eq:intro:msD3}
 \sigma^{-2} 
 =\frac{(2\xi-1)^2}{(1-\beta -\delta)\beta \delta}+\frac{4}{1-\beta -\delta}.
\end{align}
When $\xi$ and $ \beta \delta$
approach to  $\frac{1}{2}$ and $0$, respectively,
$ \sigma^{-2} $ does not take one definite value.
That is, we have
\begin{align}
\Big\{
\Big(\frac{(2\xi-1)^2}{(1-\beta -\delta)\beta \delta}+\frac{4}{1-\beta -\delta}\Big)^{-1}
\Big|\xi\to \frac{1}{2}, \beta \delta \to 0\Big\}
= 
\Big[0, \frac{1}{4}\Big].
 \end{align}
In the neighborhood
$\xi=\frac{1}{2}$ and
$ \beta \delta=0$,
$ \sigma^2$ takes all possible values in the interval $[0,1/4]$.

To see the detail of this behavior, we consider two parameterized subsets with two free parameters.
First, we fix $\xi=1/2$. In this case, $\sigma^2$ is simplified as
\begin{align}\label{eq:intro:msD4}
 \sigma^2 
 =\frac{(1-\beta -\delta)\beta \delta}{4 \beta \delta}
 =\frac{(1-\beta -\delta)}{4},
\end{align}
which takes all possible values in $[1/8,1/4]$ under the limit $ \beta \delta\to 0$.

Next, we consider the case when $Y'$ and $Z'$ are independent so that
we have two free parameters $\xi$ and $\kappa$ with the conditions
$0 \le \xi \le 1$ and $0 \le \kappa \le 1$ .
$\sigma^2$ is simplified as
\begin{align}\label{eq:intro:msD5}
 \sigma^2 
 =\frac{(1-\kappa)\kappa^2 \xi(1-\xi)}{ (2\xi-1)^2+4\kappa^2 \xi(1-\xi)}.
\end{align}
Fig. \ref{fig:sigma} plots \eqref{eq:intro:msD5}, which shows that
$ \sigma^2 $ does not take a definite value
when $\xi\to \frac{1}{2}$ and $\kappa\to 0$.

\begin{figure}[htbp]
\centering
\includegraphics[width=.96\linewidth]{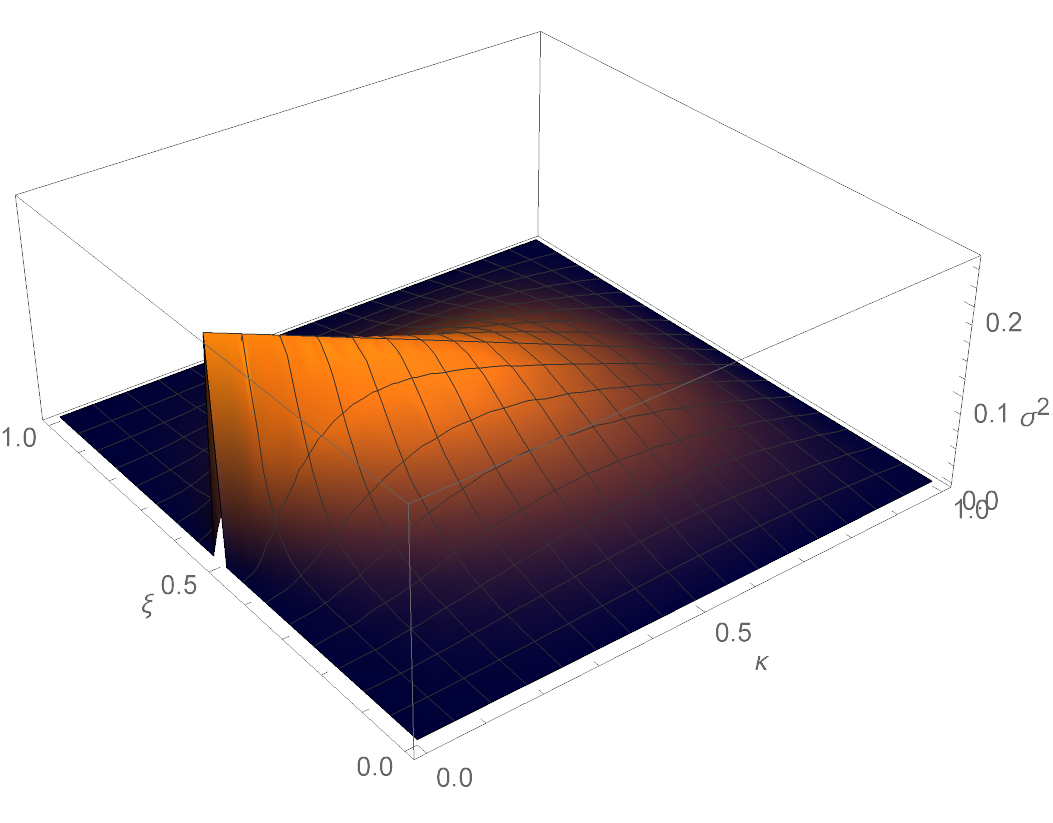}
\caption{Plots of $ \sigma^2$ as functions of $\kappa$ and $\xi$
when $\tau=0$.}
\label{fig:sigma}
\end{figure}



One of the main results is the following form of central limit theorem: 

\begin{thm}\label{thm:intro:CLT}
Assume the condition \eqref{eq:intro:II:asmp}.
Then, we have for any real $t<u$ that 
\begin{align}
 \lim_{n \to \infty} 
 \Prb_n\Bigl[t \le \frac{X-n \mu}{\sqrt{n} \sigma} \le u\Bigr]
 = \frac{1}{\sqrt{2\pi}} \int_t^u e^{-s^2/2} \, d s.
\end{align}
In particular, the pmf $p(x) \ceq p(x \midd n,m,k,l)$ behaves as 
\begin{align}\label{eq:intro:II}
&    p(x) = \Psi(x) \cdot \bigl(1+o(n^0)\bigr) \quad (n \to \infty), \quad
 \Psi(x) \ceq \frac{1}{\sqrt{2\pi n}\sigma}
 \exp\biggl(-\frac{1}{2}\Bigl(\frac{x-n \mu}{\sqrt{n} \sigma}\Bigr)^2\biggr).
\end{align}
\end{thm}

%
At this moment we do not have good understanding or interpretation
of the parameters $\mu$ and $\sigma$.
Another remark is that the quantities $\mu$ and $\sigma$ are invariant 
under the system symmetry \eqref{eq:intro:sym}.
Indeed, in terms of the fixed ratio parameters, the symmetry is written as 
$(\alpha,\beta,\gamma,\delta) \lrto (\gamma,\delta,\alpha,\beta)$,
and we see from \eqref{eq:intro:msD} that $D$ and $\sigma$, 
and hence $\mu$ are invariant under the switch.


Beyond the central limit Theorem \ref{thm:intro:CLT} in Type II limit,
we obtain the following result.

\begin{thm}
\label{thm:EV-BB}
Assume the condition \eqref{eq:intro:II:asmp}.
Then, the expectation and the variance asymptotically  behave as 
\begin{align}
 \bbE[X] &= n \mu + \phi + o(n^{0}), \quad 
 \phi \ceq \frac{\sigma^2-\mu}{1-2\mu}, \\
 \lim_{n \to \infty} \bbV \Bigl[\frac{X}{\sqrt{n}}\Bigr] &= 
 \lim_{n \to \infty} \bbE \Bigl[\Big(\frac{X-n\mu}{\sqrt{n}}\Big)^2\Bigr] 
 = \sigma^2 \ceq \frac{(\alpha+\gamma)\beta\delta}{D} \\
 \lim_{n \to \infty} \bbE \Bigl[\Big(\frac{X-n\mu}{\sqrt{n}}\Big)^j\Bigr] 
&<\infty \label{NVT}
\end{align}
where we used $\mu=(1-\sqrt{D})/2$ as before, and 
$j=3,4, \ldots.$
\end{thm}

In the special case $\alpha \gamma=0$, as shown in Section \ref{ss:II:sp},
we can further discuss the tight exponential evaluation of the tail probability as follows.
\begin{prp}\label{prp:5.6.2}
In Type II limit with $\alpha \beta \delta >0$, $\gamma=0$ 
and $0< \xi = \alpha+\beta \le 1/2$, 
the probability $\Prb_n \bigl[| \tfrac{X}{n} - \mu |\ge \epsilon \bigr]$
goes to zero exponentially for any $\epsilon >0$.
\end{prp}

Using this fact,
we can show Theorem \ref{thm:EV} in a more direct way than its proof,
i.e., without using the recurrence relation for the pmf $p(x)$.
We omit the detail, and leave it to the reader.

\subsection{Outline of proofs of the above statements
under Type II limit}
\label{sss:intro:II-3}
Let us sketch the outline of proofs of the above statements
under Type II limit.

In Section \ref{ss:AE-prf}, 
for Theorem \ref{thm:intro:E}, 
we consider the random variable $Z \ceq \frac{X-n \mu}{n}$.
Let $G_n(z)$ be the cdf for $Z$.
The key tool is the recurrence relation of the pmf $p(x)$ 
in Lemma \ref{lem:intro:rec}, and the main idea is to convert it to 
the recurrence relation of the cdf $\sum_{u=0}^x p(x)$.
Since the support of $Z$ is compact in the limit $n \to \infty$, 
we have the limit distribution $G \ceq \lim_{n \to \infty} G_n$.
Using the recursion of $p(x)$, we can prove
that the limit distribution $G$ has positive probability only at $Z=0$
(Theorem \ref{thm:LLN}), which is nothing but the law of large numbers.

For the central limit theorem, 
first, Section \ref{thm:intro:CLT} gives
an ad-hoc derivation of the normal distribution $\Psi(x)$, 
or the parameters $\mu$ and $\sigma$.
We start with the recurrence relation obtained in the paper \cite{HHY1}, which 
is stated in Lemma \ref{lem:intro:rec}, 
regarding it as a difference equation characterizing the pmf $p(x)$.
Assuming that there exists a limit distribution of $p(x)$ for $n \to \infty$ 
which is differentiable, we derive a differential equation of the limit distribution 
from the difference equation of $p(x)$ by taking $n \to \infty$. 
The obtained differential equation turns out to be the one characterizing $\Psi(x)$.

Section \ref{ss:CLT-prf} gives its rigorous proof.
Again, 
the main idea is to convert the recurrence relation of the pmf $p(x)$ to 
the recurrence relation of the cdf $\sum_{u=0}^x p(x)$.
We consider the random variable $Y \ceq \frac{X-n \mu}{\sqrt{n}}$,
whose support is non-compact in the limit $n \to \infty$. 
We denote by $F_n$ the cdf for $Y$.
Using Helly lemma (Fact \ref{fct:Helly}), we have a limit function $F$ of $F_n$, 
which might not be a probability distribution a priori.
In order to show that $F$ is actually a probability distribution,
we give a preliminary estimation of the tails (Proposition \ref{prp:tail}).
Then, we can make a detailed analysis on the the differential equation of $F$
obtained from the  recurrence relation of the pmf $p(x)$,
culminating in Lemma \ref{lem:CLT:lem4}.
The obtained differential equation turns out to be the one characterizing
the standard normal distribution, 
and we have the central limit Theorem \ref{thm:intro:CLT}.

\section{Type I Limit} 
\label{ss:I}

%

The aim of this section is the proof of Theorem \ref{thm:intro:I}.
Consider Type I limit.
Let us denote by $B_{\xi,j}(x)$ the binomial distribution $\{0,1,\dotsc,j\} \to \bbR$
with parameters $0 \le \xi \le 1$ and $j \in \bbN$
whose pmf is given by $x \mapsto \binom{j}{x} \xi^x (1-\xi)^{j-x}$, 
and by $*$ the convolution of probability distributions.

\begin{thm}\label{thm:I}
In the limit $n \to \infty$ with $k,l$ and $\xi=m/n$ fixed, we have 
\begin{align}\label{eq:I:B*B}
 P_{n,n \xi,k,l} \sim P_{R|\xi;k,l} \ceq B_{\xi,k-l} * B_{1-\xi,l}.
\end{align}
The limit distribution $P_{R|\xi,k,l}$ is a discrete probability distribution 
on $\{0,1,\dotsc,k\}$, and its expectation and variance are given by 
$(k-l)\xi+l(1-\xi)$ and $k \xi(1-\xi)$, respectively.
\end{thm}

We denote by $q(x|\xi;k,l)$ the pmf of the limit distribution $P_{R|\xi;k,l}$.
By definition of the convolution $*$, its explicit form is given by
\begin{align}\label{eq:I:q}
 q(x \midd \xi;k,l) = \xi^{l-x} (1-\xi)^{k-l-x} \sum_{u=0 \vee (x-k+l)}^{x \wedge l} 
 \binom{k-l}{x-u} \binom{l}{u} \xi^{2(x-u)} (1-\xi)^{2u}.
\end{align}
Note that the limit distribution $P_{R|\xi;k,l}$ 
inherits the system symmetry \eqref{eq:intro:sym} of the distribution $P_{n,m,k,l}$, 
and is invariant under the switch $(\xi,k-l,l) \lrto (1-\xi,l,k-l)$, 
which can be checked by the definition \eqref{eq:I:B*B}
and also by the expression \eqref{eq:I:q} of limit pmf $q(x \midd \xi;k,l)$.

Now, we show Theorem \ref{thm:I}. The proof consists of two steps.

\begin{proof}[{The first step of the proof of Theorem \ref{thm:I}}]
By the symmetry remarked after \eqref{eq:I:q},
it is enough to consider the case $m \le n-m$.
Thus, we can use Hahn presentation of the pmf $p(x)=p(x \midd n,m,k,l)$ 
in Theorem \ref{thm:Hahn}. Let us show it again: For $m \le n-m$, we have 
\begin{align}\label{eq:I:Hahn}
 p(x) = \frac{\binom{n}{x}}{\binom{n}{m}} \frac{n-2x+1}{n-x+1}
 \sum_{i \ge 0} \frac{\binom{M}{i}\binom{N}{i}}{\binom{m}{i}\binom{n-m}{i}} 
 \sum_{r = 0}^i (-1)^r \binom{x}{r} \binom{m-x}{i-r} \binom{n-m-x}{i-r}
\end{align}
with $M \ceq m-l$ and $N \ceq n-m-k+l$.
In this first step of the proof, 
we transform the formula \eqref{eq:I:Hahn} without taking limit. The result is:

\begin{prp}\label{prp:I:i-vanish}
The formula \eqref{eq:I:Hahn} is equivalent to 
\begin{align}
&p(x) = \frac{\binom{n-k}{m-l}}{\binom{n}{m}} 
 \frac{n-2x+1}{n-x+1} \frac{\ff{n}{x}}{\ff{m}{x} \ff{(n-m)}{x}} \frac{1}{x!} w(x), \\
&w(x) \ceq \sum_{\substack{0 \le u,v \le x \\ u+v \ge x}} \ff{l}{u} \ff{(k-l)}{v}
 \frac{(\ff{M}{x-u})^2 (\ff{N}{x-v})^2}{\ff{(M+N)}{2x-u-v}}
 \frac{x!}{(u+v-x)!(x-u)!(x-v)!}.
\end{align}
\end{prp}

\begin{proof}
Let us set $T \ceq \binom{n}{m}^{-1} \binom{n}{x} \frac{n-2x+1}{n-x+1}$ and 
\begin{align}
 W(x) \ceq \binom{n-k}{m-l}^{-1} \sum_{i=0}^{M \wedge N} 
 \frac{ \binom{M}{i} \binom{N}{i} }{\binom{m}{i} \binom{n-m}{i}} \, 
 \sum_{r=0}^{i \wedge x}
 (-1)^r \binom{x}{r} \binom{m-x}{i-r} \binom{n-m-x}{i-r}.
\end{align}
Thus, the formula \eqref{eq:I:Hahn} is $p(x) = T \cdot \binom{n-k}{m-l} \cdot W(x)$,
and the consequence will be obtained if we show 
\begin{align}\label{eq:i-van:rel}
 W(x) = \frac{1}{\ff{m}{x} \ff{(n-m)}{x}} \sum_{u,v=0}^x 
 \ff{l}{u} \ff{(k-l)}{v} \frac{(\ff{M}{x-u})^2 (\ff{N}{x-v})^2}{\ff{(M+N)}{2x-u-v}}
 \binom{x}{u+v-x}\binom{2x-u-v}{x-u}.
\end{align}

We start with 
\begin{align}
 \frac{\binom{m-x}{i-r}}{\binom{m}{i}} = 
 \frac{\ff{i}{r} \ff{(m-i)}{x-r}}{\ff{m}{x}},  \quad 
 \frac{\binom{n-m-x}{i-r}}{\binom{n-m}{i}} = 
 \frac{\ff{i}{r} \ff{(n-m-i)}{x-r}}{\ff{(n-m)}{x}}.
\end{align}
Thus, we have 
\begin{align}
&W(x) \ceq I(x)/\bigl(\ff{m}{x} \ff{(n-m)}{x}\bigr), \label{eq:I:W(x)} \\ 
&I(x) \ceq \binom{M+N}{N}^{-1} \sum_{i=0}^{M \wedge N} 
 \sum_{r=0}^{i \wedge x} \binom{M}{i}\binom{N}{i} 
 (-1)^r \binom{x}{r} \ff{i}{r} \ff{(m-i)}{x-r} \ff{i}{r} \ff{(n-m-i)}{x-r}.
\end{align}
Next, using the identity 
\begin{align}
 \ff{(y+z)}{j} = \sum_{U=0}^j \ff{y}{U} \ff{z}{j-U} \binom{j}{U}
\end{align}
for commuting letters $x,y$ and $j \in \bbN$, we have 
\begin{align}
\begin{split}
& \ff{i}{r} \binom{M}{i} \cdot \ff{(m-i)}{x-r} 
 =\frac{\ff{M}{i}}{(i-r)!} \cdot 
  \sum_{U=0}^{x-r} \ff{(M-i)}{U} \ff{l}{x-r-U} \binom{x-r}{U} \\
&=\sum_{U=0}^{x-r} \frac{\ff{M}{i+U}}{(i-r)!} \ff{l}{x-r-U} \binom{x-r}{U} 
 =\sum_{U=0}^{x-r} \ff{M}{r+U} \binom{M-r-U}{i-r} \ff{l}{x-r-U} \binom{x-r}{U}.
\end{split}
\end{align}
Similarly, we have 
\begin{align}
 \ff{i}{r} \binom{N}{i} \cdot \ff{(n-m-i)}{x-r} =
 \sum_{V=0}^{x-r} \ff{N}{r+V} \binom{N-r-V}{i-r} \ff{(k-l)}{x-r-V} \binom{x-r}{V}.
\end{align}
Then, changing the order of summations, we can rewrite $I(x)$ as 
\begin{align}
\begin{split}
 I(x) = &\sum_{U,V=0}^x \sum_{r=0}^x 
 (-1)^r \binom{x}{r} \ff{l}{a-r-U} \binom{x-r}{U} \ff{(k-l)}{x-r-V} \binom{x-r}{V}
 \ff{M}{r+U}  \ff{N}{r+V} \\
&\cdot \binom{M+N}{N}^{-1} \sum_{i=0}^{M \wedge N} \binom{M-r-U}{i-r} \binom{N-r-V}{i-r}.
\end{split}
\end{align}
Using the Chu-Vandermonde formula 
\begin{align}\label{eq:intro:CV2}
 \sum_{r \ge 0} \binom{A}{a-r} \binom{B}{r} = \binom{A+B}{a},
\end{align}
we have
\begin{align}
\begin{split}
 \sum_{i=0}^{M \wedge N} \binom{M-r-U}{i-r} \binom{N-r-V}{i-r} 
&= \binom{M+N-2r-U-V}{N-r-V} \\
&= \binom{M+N}{N} \frac{\ff{M}{r+U} \ff{N}{r+V}}{\ff{(M+N)}{2r+U+V}}.
\end{split}
\end{align}
Replacing $u \ceq x-r-U$ and $v \ceq x-r-V$, we can further rewrite $I(x)$ as 
\begin{align}\label{eq:I:IJ}
\begin{split}
&I(x) = \sum_{u,v=0}^x \ff{l}{u} \ff{(k-l)}{v}
 \frac{(\ff{M}{x-u})^2 (\ff{N}{a-v})^2}{\ff{(M+N)}{2x-u-v}} J(u,v), \\
&J(u,v) \ceq \sum_{r=0}^x (-1)^r \binom{x}{r} \binom{x-r}{u} \binom{x-r}{v}.
\end{split}
\end{align}

By \eqref{eq:I:W(x)} and \eqref{eq:I:IJ}, 
in order to prove \eqref{eq:i-van:rel}, it remains to show
\begin{align}\label{eq:I:h(u,v)}
 J(u,v) = \begin{cases} 0        & (u+v  <  x) \\ 
 \frac{x!}{(u+v-x)!(x-u)!(x-v)!} & (u+v \ge x)\end{cases}.
\end{align}
Note that $J(u,v)$ is equal to the coefficient of $A^u B^v$ 
in $\bigl((1+A)(1+B)-1)^x$ for commuting letters $A$ and $B$.
Thus, it is also equal to the coefficient of $(A B)^{u+v-x} A^{x-v} B^{x-u}$ 
in $(A B + A + B)^x$, and we have the above equality \eqref{eq:I:h(u,v)}.
\end{proof}
This is the end of the first step.
\end{proof}

\begin{proof}[{The second step of the proof of Theorem \ref{thm:I}}]
Next, as the second step, 
we take Type I limit $n \to \infty$ of Proposition \ref{prp:I:i-vanish}.

\begin{prp}\label{prp:I:lim}
In the limit $n \to \infty$ with $k,l$ and $\xi=m/n$ fixed,
the function $h(x)$ in Proposition \ref{prp:I:i-vanish} behaves as 
\begin{align}
 \frac{\ff{n}{x}}{\ff{m}{x} \ff{(n-m)}{x}} \frac{1}{x!} h(x) \sim 
 \xi^{-x} (1-\xi)^{-x} \sum_{u=0}^x \binom{l}{u} \binom{k-l}{x-u} 
 \xi^{2(x-u)} (1-\xi)^{2u} \quad (n \to \infty).
\end{align}
\end{prp}

\begin{proof}
We find from the expression of $h(x)$ that 
the terms with $u+v=x$ are dominant in the limit.
Then, we can estimate the left hand side of the statement as
\begin{align}
\begin{split}
 \frac{\ff{n}{x}}{\ff{m}{x} \ff{(n-m)}{x}} \frac{1}{x!} h(x) \sim 
&\frac{n^x}{m^x (n-m)^x} \frac{1}{x!} 
 \sum_{\substack{0 \le u,v \le x \\ u+v = x}} \ff{l}{u} \ff{(k-l)}{v} 
 \frac{m^{2(x-u)} (n-m)^{2(x-v)}}{n^{2x-u-v}} \binom{x}{u} \\
&= \xi^{-x} (1-\xi)^{-x} 
    \sum_{u=0}^x \binom{l}{u} \binom{k-l}{u} \xi^{2(x-u)} (1-\xi)^{2u}.
\end{split}
\end{align}
\end{proof}

By Proposition \ref{prp:I:lim} and the estimation
$\frac{\ff{m}{l} \ff{(n-m)}{k-l}}{\ff{n}{k}} \frac{n-2x+1}{n-x+1} 
 \sim \xi^l (1-\xi)^{k-l}$,
we have 
\begin{align}
 p(x) \sim
 \xi^{l-x} (1-\xi)^{k-l-x} \sum_{u=0}^x 
 \binom{k-l}{x-u} \binom{l}{u} \xi^{2(x-u)} (1-\xi)^{2u}, 
\end{align}
which is equivalent to the expression \eqref{eq:I:q} of $q(x \midd \xi;k,l)$.
Thus, we obtain the first half \eqref{eq:I:B*B} of Theorem \ref{thm:I}.

Next we show the latter half of Theorem \ref{thm:I}, i.e, 
calculate the expectation and the variance of the limit distribution $P_{R|\xi,k,l}$.
For that, we consider the cumulant generating function for $P_{R|\xi;k,l}$:
\begin{align}
 K(s) \ceq \log \sum_{x=0}^k e^{s x} q(x \midd \xi;k,l).
\end{align}
Recall that, denoting by $K_P$ the cumulant generating function 
of a probability distribution $P$, we have $K_{P_1 * P_2} = K_{P_1} + K_{P_2}$. 
Then, since $K_{B_{\xi,j}}(s) = j \log(1-\xi+\xi e^s)$ 
for the binomial distribution $B_{\xi,j}$, we have
\begin{align}
 K(s) = (k-l) \log \bigl(1-\xi+\xi e^s \bigr) + l \log \bigl(\xi+(1-\xi)e^s \bigr).
\end{align} 
Hence, the expectation and the variance of $P_{R|\xi;k,l}$ are given by 
\begin{align}
 \frac{d}{d s}    K(0) = (k-l)\xi+l(1-\xi), \quad 
 \frac{d^2}{d^2 s}K(0) = k \xi(1-\xi).
\end{align}
Thus, we finished the proof of Theorem \ref{thm:I}
\end{proof}

\section{Ad hoc derivation of central limit theorem Theorem \ref{thm:intro:CLT}}
\label{sss:II:gen}

Now, we assume the condition \eqref{eq:intro:II:asmp}
and consider the Type II limit,
i.e., $n \to \infty$ with fixed ratios $\alpha,\beta,\gamma,\delta$
satisfying the condition \eqref{eq:intro:II:asmp}.
Let us write it again: 
Recall that 
the quantity $\sigma$ is well-defined and positive.
As we will see below, it describes the variance of the limit distribution.

Let us define the normal distribution $\Psi(x)$ by
\begin{align}\label{eq:II:Psi}
 \Psi(x) \ceq \frac{1}{\sqrt{2\pi n}\sigma}
  \exp\biggl(-\frac{1}{2}\Bigl(\frac{x-n \mu}{\sqrt{n} \sigma}\Bigr)^2\biggr).
\end{align}
Remember that the parameters $\mu$ and $\sigma$ 
have the following forms
\begin{align}\label{eq:II:msD}
    \mu  = \frac{1-\sqrt{D}}{2}, \quad
 \sigma  = \sqrt{\frac{(\alpha+\gamma)\beta\delta}{D}}, \quad 
      D  = 1-4(\alpha\gamma+\alpha\delta+\beta\gamma). 
\end{align}
Also, we additionally employ the quantity 
\begin{align}\label{eq:II:msD2}
\nu \ceq \frac{1+\sqrt{D}}{2}.
\end{align}
Note that the assumption \eqref{eq:intro:II:asmp} guarantees 
\begin{align}
 D>0, \quad \sigma>0,\quad  0<\mu<\xi \wedge \kappa \wedge (1-\kappa)
\end{align}
with $\xi \ceq \alpha+\beta$ and $\kappa \ceq \alpha+\gamma$, and 
the normal distribution $\Psi(x)$ in \eqref{eq:II:Psi} makes sense.

Below we explain an ad-hoc argument to guess that 
$p(x)$ behaves in $n \to \infty$ as 
\begin{align}\label{eq:II:P}
  p(x) \dapp \Psi(x). 
\end{align}
The outline is as follows:
Assuming that there exists a limit distribution of $p(x)$ in $n \to \infty$,
we will derive a differential equation for the limit distribution in an ad-hoc way,
by taking $n \to \infty$ in the recurrence formula of $p(x)$ of Lemma \ref{lem:intro:rec}.
The obtained differential equation turns out to be the one determining 
the normal distribution $\Psi(x)$.

\begin{proof}[Ad-hoc derivation of \eqref{eq:II:P}] 
Let us assume $m \le n-m$, i.e., $\xi \le 1/2$, until \eqref{eq:II:ah-Phi}.
We then have the recurrence formula of $p(x)$ in Lemma \ref{lem:intro:rec},
which can rewritten  by a direct computation as
\begin{align}\label{eq:rec:0}
\begin{split}
& \frac{n-2x-1}{n-2x+1} a_x 
  \left(\frac{x+1}{n-x} \, p(x+1)-\frac{x}{n-x+1} \, p(x) \right) \\
&-\frac{n-2x+3}{n-2x+1} \frac{n-x}{x-1} c_x 
  \left(p(x)-\frac{n-x+1}{n-x+2} \, p(x-1) \right) \\
& = \frac{n-x}{n-x+1} 
  \left(\left(1-\frac{n-2x-1}{n-2x+1}\frac{  x}{n-x}\right)a_x + 
        \left(1-\frac{n-2x+3}{n-2x+1}\frac{n-x}{x-1}\right)c_x - M N \right) p(x)
\end{split}
\end{align}
with $M \ceq m-l$ and $N \ceq n-m-k+l$.
Let us consider Type II limit $n \to \infty$ of this recursion, i.e., 
fixing $\xi=m/n$, $\kappa=k/n$ and $\alpha=l/n$.
Denoting $t \ceq x/n$, we deduce from the above \eqref{eq:rec:0} that 
\begin{align}\label{eq:II:rec}
\begin{split}
   \frac{t}{1-t} a^{(\infty)}_t \bigl(p(x+1)-p(x)\bigr) 
&- \frac{1-t}{t} c^{(\infty)}_t \bigl(p(x)-p(x-1)\bigr) \\
&\sim
 \Bigl( a^{(\infty)}_t \frac{1-2t}{1-t} + c^{(\infty)}_t \frac{2t-1}{t} -
        \beta \delta \Bigr) p(x),
\end{split}
\end{align}
where we set $a^{(\infty)}_t \ceq \lim_{n \to \infty} a_x/n^2$ and 
$c^{(\infty)}_t \ceq \lim_{n \to \infty} c_x/n^2$.
We also used the symbol $F \sim G$ to mean $F-G=O(1/n)$, 
as in the previous \S \ref{ss:I}.
The explicit forms of the coefficients are given by
\begin{align}
 a^{(\infty)}_t = \frac{(\xi-t)(1-\xi-t)}{(1-2t)^2} (1-t)(1-\kappa-t), \quad 
 c^{(\infty)}_t = \frac{(\xi-t)(1-\xi-t)}{(1-2t)^2} t(t-\kappa).
\end{align}
As for the right hand side of \eqref{eq:II:rec}, we have 
\begin{align}
 \frac{1-2t}{1-t} a^{(\infty)}_t + \frac{2t-1}{t} c^{(\infty)}_t
  = (\xi-t)(1-\xi-t),
\end{align}
which yields 
\begin{align}
\begin{split}
(\text{RHS of \eqref{eq:II:rec}}) = (t-\mu)(t-\nu) p(x),.
\end{split}
\label{eq:II:quad}
\end{align}
Note that $\mu$ and $D$ are the same as those in \eqref{eq:II:msD}.
Now, we consider $p(x)$ as a function of $t=x/n$ and denote it by 
$\varpi(t) \ceq p(n t)=p(x)$.
Assume that we can approximate 
\begin{align}\label{eq:varpi'}
\varpi'(t) \approx  n \bigl(p(x+1)-p(x)\bigr) \approx n \bigl(p(x)-p(x-1)\bigr).
\end{align} 
Here and hereafter, the symbol $\approx$ indicates this assumption.
Then, we have from \eqref{eq:II:rec} that   
\begin{align}
 \frac{\varpi'(t)}{\varpi(t)} \approx
 \frac{n (t-\mu)(t-\nu)}{a^{(\infty)}_t \frac{t}{1-t}
 - c^{(\infty)}_t \frac{1-t}{t}}
 = \frac{n (t-\mu)(t-\nu) (1-2t)}{\kappa (\xi-t)(1-\xi-t)}.
\end{align}
Let us rewrite this approximate equation as 
\begin{align}\label{eq:II:varpi}
 \frac{d}{d t} \log \varpi(t) \approx -n(t-\mu) \cdot R(t), \quad 
 R(t) \ceq \frac{(\nu-t) (1-2t)}{\kappa (\xi-t)(1-\xi-t)}.
\end{align}

Now, using the quantity $\sigma$ in \eqref{eq:II:msD},
we consider the variable 
\begin{align}
 y \ceq \frac{x- n \mu}{\sqrt{n} \sigma} = \frac{\sqrt{n}(t-\mu)}{\sigma}.
\end{align}
Then, in the limit $n \to \infty$, 
the approximate equation \eqref{eq:II:varpi} has the form
\begin{align}\label{eq:II:R}
 \frac{d}{d y} \log \varpi(t) \approx -y \sigma^2 R(\mu) = 
 -y \sigma^2 \frac{(\nu-\mu)(1-2\mu)}{\kappa (\xi-\mu)(1-\xi-\mu)} = -y.
\end{align}
In the last equality, we used $\nu-\mu=1-2\mu=\sqrt{D}$,
$\xi=\alpha+\beta$, $\kappa=\alpha+\gamma$ and 
$(\xi-\mu)(1-\xi-\mu)=\xi(1-\xi)-(1-D)/4=\beta \delta$ to obtain
\begin{align}\label{eq:II:sigma^2}
 \frac{(\nu-\mu)(1-2\mu)}{\kappa(\xi-\mu)(1-\xi-\mu)}
 = \frac{D}{\kappa \beta \delta} = \sigma^{-2}.
\end{align}
Thus we obtained from \eqref{eq:II:R} the differential equation
\begin{align}\label{eq:II:ode}
 \frac{d}{d y} \log P(y) =  -y
\end{align}
for a function $P(y)$ of the variable $y$,
which characterizes the normal distribution 
\begin{align}\label{eq:II:ah-Phi}
 P(y) = \frac{1}{\sqrt{2\pi n}\sigma} \exp\biggl(-\frac{y^2}{2}\biggr)
 = \frac{1}{\sqrt{2\pi n}\sigma} \exp\biggl(
  -\frac{1}{2}\Bigl(\frac{x-n \mu}{\sqrt{n} \sigma}\Bigr)^2\biggr) =: \Psi(x).
\end{align}
Hence, we can guess $p(x) \approx \Psi(x)$ in the limit $n \to \infty$.

Finally, since the quantities $\mu$ and $\sigma$ are invariant 
under the system symmetry \eqref{eq:intro:sym} (see Remark \ref{rmk:msD:sym} below),
we can remove the assumption $\xi \le 1/2$ put in the beginning.
\end{proof}

\begin{rmk}\label{rmk:msD:sym}
The quantities $\mu$ and $\sigma$ are invariant 
under the system symmetry \eqref{eq:intro:sym}.
Indeed, in terms of the fixed ratio parameters $\alpha,\beta,\gamma,\delta$, 
the symmetry is written as 
\begin{align}
 (\alpha,\beta,\gamma,\delta) \llrto (\gamma,\delta,\alpha,\beta),
\end{align}
and we see from \eqref{eq:II:msD} that $D$ and $\sigma$, 
and hence $\mu$ are invariant under the switch.
\end{rmk}

Although \eqref{eq:II:P} is an ad-hoc observation,
we can actually prove the following form of the central limit theorem for $p(x)$, 
Theorem \ref{thm:intro:CLT}. 

\section{Proof of law of large numbers, Theorem \ref{thm:intro:E}}\label{ss:AE-prf}


In this subsection, we give a proof of the law of large numbers for $X/n$
(Theorem \ref{thm:intro:E}).
Hereafter the word ``limit" means Type II limit $n \to \infty$
with the fixed ratios $\xi=\frac{m}{n}$, $\kappa=\frac{k}{n}$, $\alpha=\frac{l}{n}$, 
$\beta=\frac{m-l}{n}$, $\gamma=\frac{k-l}{n}$ and $\delta=\frac{n-m-k+l}{n}$.
We will also use the symbol $p_n(x) = p(x \midd n,m,k,l)$.

\subsection{Recursion formula}\label{sss:AE:rec}

Assume $\xi = \frac{m}{n} \le \frac{1}{2}$.
Then, the range of the random variable $X$ is $[0,x_{\tmx}]$ with 
$x_{\tmx} \ceq m \wedge k$, and we have the recurrence relation of the pmf 
$p_n(x)=p(x \midd n,m,k,l)$ in Lemma \ref{lem:intro:rec}.
The recursion can be rewritten as 
\begin{align}\label{eq:AE:NAP}
 A_x^{(n)} p_n(x-1) + B_x^{(n)} p_n(x) = C_x^{(n)} p_n(x+1) \quad (0<x<m)
\end{align}
with coefficients given by 
\begin{align}
\begin{split}\label{eq:AE:Ax} 
  A_x^{(n)} \ceq 
  &-\frac{c_x}{n^2} \frac{n-2x+3}{n-x+2} \frac{\binom{n}{x}}{\binom{n}{x-1}}  \\
 =& \frac{(\frac{m}{n}-\frac{x}{n}+\frac{1}{n})(1-\frac{m}{n}-\frac{x}{n}+\frac{1}{n})
          (\frac{k}{n}+\frac{1}{n}-\frac{x}{n})
          (1-\frac{x}{n}+\frac{1}{n})(1-\frac{2x}{n}+\frac{3}{n})}
         {(1-\frac{x}{n}+\frac{2}{n})
          (1-\frac{2x}{n}+\frac{1}{n})(1-\frac{2x}{n}+\frac{2}{n})}, 
\end{split} 
\\
\begin{split} \label{eq:AE:Bx} 
 B_x^{(n)} \ceq & \, \frac{a_x+c_x-M N}{n^2} \frac{n-2x+1}{n-x+1} \\
 =&     \frac{1-\frac{2x}{n}+\frac{1}{n}}{1-\frac{x}{n}+\frac{1}{n}} 
  \biggl(\frac{(  \frac{m}{n}-\frac{x}{n})(1-\frac{m}{n}-\frac{x}{n})
               (1-\frac{k}{n}-\frac{x}{n})(1-\frac{x}{n}+\frac{1}{n})}
              {(1-\frac{2x}{n})(1-\frac{2x}{n}+\frac{1}{n})} \\
 &\hspace{5em}
        -\frac{(  \frac{m}{n}-\frac{x}{n}+\frac{1}{n})
               (1-\frac{m}{n}-\frac{x}{n}+\frac{1}{n})
               (  \frac{k}{n}+\frac{1}{n}-\frac{x}{n})\frac{x}{n}}
              {(1-\frac{2x}{n}+\frac{1}{n})(1-\frac{2x}{n}+\frac{2}{n})}
        -\frac{M N}{n^2} \biggr), 
\end{split}
\\
\begin{split}\label{eq:AE:Cx}
 C_x^{(n)} \ceq & \frac{a_x}{n^2} \frac{n-2x-1}{n-x} \frac{\binom{n}{x}}{\binom{n}{x+1}} \\
 =& \frac{(  \frac{ m}{n}-\frac{x}{n})(1-\frac{m}{n}-\frac{x}{n})
            (1-\frac{ k}{n}-\frac{x}{n})(  \frac{x}{n}+\frac{1}{n})
            (1-\frac{x}{n}+\frac{1}{n})(1-\frac{2x}{n}-\frac{1}{n})}
           {(1-\frac{2x}{n})(1-\frac{2x}{n}+\frac{1}{n})(1-\frac{x}{n})^2}.
\end{split}
\end{align}
In the expression of $B_x^{(n)}$,  we used $M \ceq m-l$ and $N \ceq n-m-k+l$.
Let us record:

\begin{lem}\label{lem:AE:AxCx}
$A_x^{(n)}$ is positive for any $x \in [0,x_{\tmx}]$.
$C_x^{(n)}$ is positive, zero and negative 
according to $x<n-k$, $x=n-k$ and $x>n-k$, respectively.
\end{lem}

Let us use the variable $t=x/n$ and 
the fixed ratios $\xi=m/n$, $\kappa=k/n$ and $\alpha=l/n$.
Then, as we saw in \eqref{eq:II:rec}, 
the coefficients $A_x^{(n)}$, $B_x^{(n)}$ and $C_x^{(n)}$ converge to 
\begin{align}
\begin{split}
 \lim_{n \to \infty}A_{n t}^{(n)} &= \frac{(\xi-t)(1-\xi-t)(\kappa-t)}{1-2t}, \\
 \lim_{n \to \infty}C_{n t}^{(n)} &= 
 \frac{(\xi-t)(1-\xi-t)(1-\kappa-t)t}{(1-t)(1-2t)}, \\
 \lim_{n \to \infty}B_{n t}^{(n)} &= \frac{1-2t}{1-t}
 \biggl(\frac{(\xi-t)(1-\xi-t)(1-\kappa-t)(1-t)}{(1-2t)^2}  \\
&\phantom{1-2t T (} 
       -\frac{(\xi-t)(1-\xi-t)(\kappa-t)t}{(1-2t)^2} 
       -(\xi-\alpha)(1-\xi-\kappa+\alpha) \biggr), 
\end{split}
\end{align}
By \eqref{eq:II:quad}, we also have
\begin{align}\label{eq:AE:NM-lim}
 \lim_{n \to \infty}\left(A_{n t}^{(n)}+B_{n t}^{(n)}-C_{n t}^{(n)}\right)
 = \frac{1-2t}{1-t} (t-\mu)(t-\nu)
\end{align}
with $\mu$ and $\nu$ given by $\mu \ceq \frac{1-\sqrt{D}}{2}$,  
$\nu \ceq \frac{1+\sqrt{D}}{2}$ and $D \ceq 1-4(\alpha\gamma+\alpha\delta+\beta\gamma)$
as in \eqref{eq:II:msD}.

In the later discussion, we will also use 
the following rewritten form of \eqref{eq:AE:NAP}:
\begin{align}\label{eq:AE:AP}
 \bigl(p_n(x)-p_n(x+1)\bigr) -\eta_{1,x} \bigl(p_n(x-1)-p_n(x)\bigr) = 
-\eta_{2,x} p_n(x)
\end{align}
with coefficients given by 
\begin{align}
\begin{split}\label{eq:AE:eta1}
 \eta_{1,x} \ceq & -\frac{A_x^{(n)}}{C_x^{(n)}} \\
 =&-\frac{(1-\frac{x}{n})^2 (\frac{k}{n}-\frac{x}{n}+\frac{1}{n})
          (\frac{m}{n}-\frac{x}{n}+\frac{1}{n})(1-\frac{m}{n}-\frac{x}{n}+\frac{1}{n})}
         {(\frac{m}{n}-\frac{x}{n})(1-\frac{m}{n}-\frac{x}{n})(1-\frac{k}{n}-\frac{x}{n})
          (\frac{x}{n}+\frac{1}{n})(1-\frac{x}{n}+\frac{2}{n})} 
    \frac{(1-\frac{2x}{n})(1-\frac{2x}{n}+\frac{3}{n})}
         {(1-\frac{2x}{n}+\frac{1}{n})(1-\frac{2x}{n}+\frac{2}{n})}, 
\end{split}
\\
\label{eq:AE:eta2} 
 \eta_{2,x} 
\ceq&\frac{A_x^{(n)}+B_x^{(n)}-C_x^{(n)}}{C_x^{(n)}}.
\end{align}
Here we assumed $x \neq n-k$ so that $C_x^{(n)} \neq 0$ (see Lemma \ref{lem:AE:AxCx})
and $\eta_{1,x}, \eta_{2,x}$ are well-defined.
In the limit $n \to \infty$, we have 
\begin{align}\label{eq:AE:eta-lim}
 \lim_{n \to \infty} \eta_{1,n t}
 =\frac{(1- t)(t-\kappa)}{ t(1-\kappa-t)}, \quad 
 \lim_{n \to \infty} \eta_{2,n t}
 =\frac{(t-\mu)(t-\nu)(1-2t)^2}{(\xi-t)(1-\xi-t)t(1-\kappa-t)}.
\end{align}
We will also use the solutions $\zeta=\zeta_{1,x},\zeta_{2,x}$ of 
the characteristic equation
\begin{align}\label{eq:AE:ch}
 \zeta^2 - (1+\eta_{1,x}+\eta_{2,x}) \zeta + \eta_{1,x} = 0
\end{align}
for the recursion \eqref{eq:AE:AP}.
We can further rewrite \eqref{eq:AE:ch} as 
\begin{align}\label{eq:AE:AP2}
 p_n(x+1) - \zeta_{i,x} \, p_n(x) = 
 \zeta_{j,x} \bigl(p_n(x)- \zeta_{i,x} \, p_n(x-1)\bigr), \quad \{i,j\}=\{1,2\}.
\end{align}
Our choice of $\zeta_{1,x}$ and $\zeta_{2,x}$ is
\begin{align}\label{eq:AE:zeta}
 \zeta_{1,x} \ceq \frac{\theta_x-\sqrt{\theta_x^2-4\eta_{1,x}}}{2}, \quad 
 \zeta_{2,x} \ceq \frac{\theta_x+\sqrt{\theta_x^2-4\eta_{1,x}}}{2}, \quad 
 \theta_x    \ceq 1+\eta_{1,x}+\eta_{2,x}.
\end{align}

\subsection{Random variable $Z$ with bounded support}
As in \S \ref{sss:AE:rec}, we assume $\xi \le 1/2$, 
so that the range of $X/n$ is $[0,t_{\tmx}]$ with $t_{\tmx} \ceq \xi \wedge \kappa$.
We also continue to use $p_n(x) \ceq p(x \midd n,m,k,l)$.
We consider the random variable $Z \ceq \frac{X-n \mu}{n}$ which has a bounded support: 
\begin{align}
 p_n\bigl(n(z+\mu)\bigr)=0 \ \text{ unless } \ z \in [-\mu,t_{\tmx}-\mu].
\end{align}
In this part we show:

\begin{prp}\label{prp:AE:lim-p}
For any $z \in [-\mu,t_{\tmx}-\mu] \setminus \{0\}$ and any $i \in \bbZ$, we have 
\begin{align}
 \lim_{n \to \infty} p_n\bigl(n(z+\mu)-i\bigr) = 0.
\end{align}
\end{prp}

In the proof, we use the recursion \eqref{eq:AE:AP} 
with coefficients $\eta_{1,x},\eta_{2,x}$, 
and use the solutions $\zeta_{1,x}, \zeta_{2,x}$ 
in \eqref{eq:AE:zeta} of the characteristic equation \eqref{eq:AE:ch}.
We denote the limits of coefficients and solutions by
\begin{align}
 \ol{ \eta}_{i,z} \ceq \lim_{n \to \infty}  \eta_{i, n(z+\mu)}, \quad 
 \ol{\zeta}_{i,z} \ceq \lim_{n \to \infty} \zeta_{i, n(z+\mu)}  \quad (i=0,1)
\end{align}
if they converge.
In such a case, we have similar relations as \eqref{eq:AE:zeta}:
\begin{align}\label{eq:AE:zeta-lim}
 \ol{\zeta}_{1,z} = 
 \frac{\ol{\theta}_z-\bigl(\ol{\theta}_z^2-4\ol{\eta}_{1,z}\bigr)^{1/2}}{2}, \quad 
 \ol{\zeta}_{2,z} = 
 \frac{\ol{\theta}_z+\bigl(\ol{\theta}_z^2-4\ol{\eta}_{1,z}\bigr)^{1/2}}{2}, \quad 
 \ol{\theta}_z \ceq 1+\ol{\eta}_{1,z}+\ol{\eta}_{2,z}.
\end{align}

\begin{lem}\label{lem:AE:lem1}
The limits $\ol{\zeta}_{1,z}$ and $\ol{\zeta}_{2,z}$ satisfy the followings.
\begin{enumerate}
\item
For $-\mu < z < 0$, we have $\ol{\zeta}_{1,z}<0$ and $1<\ol{\zeta}_{2,z}$.

\item 
For $z=0$, we have $\ol{\zeta}_{1,0}<0$ and $\ol{\zeta}_{2,0}=1$.

\item
For $0 < z \le t_{\tmx}-\mu$, we have 
\begin{align}
 \begin{cases}
    \ol{\zeta}_{1,z}<0<\ol{\zeta}_{2,z}<1 & (0 < z < 1-\kappa-\mu) \\
  0<\ol{\zeta}_{1,z}<1<\ol{\zeta}_{2,z}   & (1-\kappa-\mu < z \le t_{\tmx}-\mu)
 \end{cases}.
\end{align}

\end{enumerate}
\end{lem}

\begin{proof}
Note that we always have $\mu \le 1-\kappa$. 
We set $f_z(\ol{\zeta}) \ceq \ol{\zeta}^2-\ol{\theta}_z \ol{\zeta} + \ol{\eta}_{1,z}$
with $\ol{\theta}_z \ceq 1+\ol{\eta}_{1,z}+\ol{\eta}_{2,z}$.
Note that $\ol{\zeta}_{1,z}, \ol{\zeta}_{2,z}$ are solutions of $f_z(\ol{\zeta})=0$,
and that $f_z(0)=\ol{\eta}_{1,z}$, $f_z(1)=-\ol{\eta}_{2,z}$.
\begin{enumerate}
\item
We find from \eqref{eq:AE:eta-lim} that $\ol{\eta}_{1,z}<0$ for $z<0 \le 1-\kappa-\mu$,
and that $\ol{\eta}_{2,z}>0$ for $-\mu<z<0$.
Hence, we have $f_z(0)<0$ and $f_z(1)<0$.
On the other hand, the inequality $\ol{\eta}_{1,z}<0$ and \eqref{eq:AE:zeta-lim} 
yield $\ol{\zeta}_{1,z}<0<\ol{\zeta}_{2,z}$.
Combining these inequalities and considering the graph of the quadratic function $f_z$,
we have the result. 


\item 
By \eqref{eq:AE:eta-lim}, we have $\ol{\eta}_{1,0}<0$ and $\ol{\eta}_{2,0}=0$.
Then, $f_0(\ol{\zeta})=(\ol{\zeta}-\ol{\eta}_{1,0})(\ol{\zeta}-1)$, and we have 
$\ol{\zeta}_{1,0}=\ol{\eta}_{1,0}<0$, $\ol{\zeta}_{2,0}=1$.

\item
If $0<z<1-\kappa-\mu$, then we have $\ol{\eta}_{1,z}, \ol{\eta}_{2,z}<0$ 
by \eqref{eq:AE:eta-lim}.
Thus, we have $f_z(0)<0<f_z(1)$, which implies the result.
Similarly, if $z>1-\kappa-\mu$, then we have $\ol{\eta}_{1,z}, \ol{\eta}_{2,z}>0$,
which means $f_z(1)<0<f_z(0)$, and we have the result.
\end{enumerate}
\end{proof}

Now, we show Proposition \ref{prp:AE:lim-p}.
\begin{proof}[{Proof of Proposition \ref{prp:AE:lim-p}}]
Let us take $z \in [-\mu,t_{\tmx}-\mu]\setminus\{0\}$ and $i \in \bbZ$ 
as in the statement.

First, we show the statement in the case $-\mu<z<0$.
We define $a_j \ceq \lim_{n \to \infty}p_n(n(z+\mu)-i+j)$ for $j=-1,0$ if the limits exist.
If not, then we take a subsequence $\{n_h\}_h$ such that the limits exist.
Note that such a subsequence does exist.
Then, for each $j \in \bbZ_{\ge 1}$, 
the limit $a_j \ceq \lim_{n \to \infty}p_n(n(z+\mu)-i+j)$ exists, 
and the recursion \eqref{eq:AE:AP2} yields
\begin{align}\label{eq:AE:lim-p:z<0}
 \ol{\zeta}_{2,z}^{-j} \bigl(a_{j+1}-\ol{\zeta}_{1,z} a_j\bigr)
 = a_0 - \ol{\zeta}_{1,z} a_{-1}
\end{align}
for any $j \in \bbN$.
By Lemma \ref{lem:AE:lem1}, we have $\ol{\zeta}_{1,z}<0<1<\ol{\zeta}_{2,z}$.
Then, the equality \eqref{eq:AE:lim-p:z<0} and $0 \le a_j \le 1$ yield 
\begin{align}\label{eq:AE:5329}
 \ol{\zeta}_{2,z}^{-j} \bigl(1+\abs{\ol{\zeta}_{1,z}} \bigr)
 \ge  a_0 + \abs{\ol{\zeta}_{1,z}} a_{-1}
\end{align}
for any $j \in \bbN$, which implies $a_{-1}=a_0=0$.
Thus, the desired claim $a_0= \lim_{n \to \infty} p_n(n(z+\mu)-i)=0$ is proved.

The second case $0<z<1-\kappa-\mu$ can be treated quite similarly.
We can define the limit $a_j \ceq \lim_{n \to \infty} p_n(n(z+\mu)-i-j)$ 
for $j \in \bbZ_{\ge -1}$, and the recursion \eqref{eq:AE:AP2} yields
$\ol{\zeta}_{2,z}^j \bigl(a_{j+1}-\ol{\zeta}_{1,z} a_j\bigr)
 = a_0 - \ol{\zeta}_{1,z} a_{-1}$ for $j \in \bbN$.
By Lemma \ref{lem:AE:lem1}, 
we have $\ol{\zeta}_{1,z}<0<\ol{\zeta}_{2,z}<1$, which yields
\begin{align}
 \ol{\zeta}_{2,z}^j \bigl(1+\abs{\ol{\zeta}_{1,z}} \bigr)
 \ge  a_0 + \abs{\ol{\zeta}_{1,z}} a_{-1}
\end{align}
for any $j \in \bbN$.
It implies $a_{-1}=a_0=0$, and we are done.

The next case is $1-\kappa-\mu<z$.
As before, we can define the limit 
$a_j \ceq \lim_{n \to \infty}p_n(n(z+\mu)-i+j)$ for $j \in \bbZ_{\ge -1}$, 
and the recursion \eqref{eq:AE:AP2} yields
the same equality as \eqref{eq:AE:lim-p:z<0}.
Then, $0 \le a_j \le 1$ yields the same inequality as \eqref{eq:AE:5329}.
By Lemma \ref{lem:AE:lem1}, we have $0<\ol{\zeta}_{1,z}<1<\ol{\zeta}_{2,z}$,
and the obtained inequality implies $a_{-1}=a_0=0$.

Next, we consider the case $z=1-\kappa-\mu$, i.e., the case $x=n-k$.
We can define the limit $a_j \ceq \lim_{n \to \infty}p_n(n(z+\mu)-i-j)$ for $j \in \bbN$.
Since \eqref{eq:AE:AP2} is not available, 
we use the original recurrence relation \eqref{eq:AE:NAP}.
By \eqref{eq:AE:Cx} and \eqref{eq:AE:NM-lim},
the coefficients are $A_{n-k}^{(n)}>0$, $B_{n-k}^{(n)}<0$ and $C_{n-k}^{(n)}=0$
satisfying $\ol{\zeta}_{2,z} \ceq -A_{n-k}^{(n)}/B_{n-k}^{(n)}<1$.
We then have $\ol{\zeta}_{2,z}^j a_j =a_0$, which yields the result $a_0=0$.

The remaining case $z=-\mu$, i.e., the case $x=0$, 
can be treated similarly as in the previous case $z=1-\kappa-\mu$.
In this case, we can also show $\lim_{n \to \infty}p_n(i)=0$ for any fixed $i \in \bbN$
directly from the explicit presentations.
We omit the detail.
\end{proof}

\subsection{Final step of our proof}\label{sss:E-prf}

We keep the notations, and consider the random variable $Z=\frac{X-n \mu}{n}$.
Let $G_n$ be the corresponding cdf, i.e., 
\begin{align}\label{eq:AE:G_n}
 G_n(z) \ceq \sum_{u=0}^{n(z+\mu)} p(u \midd n,m,k,l).
\end{align}
By the compact support property, there exists a subsequence 
$\{n_i \}_i \subset \bbZ_{>0}$ which has the limit distribution 
\begin{align}
 G \ceq \lim_{i \to \infty} G_{n_i}.
\end{align} 
Then, Theorem \ref{thm:intro:E} follows from the next statement:

\begin{thm}\label{thm:LLN}
The limit distribution $G$ has a positive probability only at $Z=0$.
\end{thm}

\begin{proof}
Let us assume $m \le n-m$, i.e., $\xi \le 1/2$, for a while.
Then, we can use the recurrence relation \eqref{eq:AE:NAP}.
Hereafter until \eqref{eq:AE:39}, we assume $z<0$.
Then, by \eqref{eq:AE:NM-lim}, we have
\begin{align}\label{eq:AE:NM}
 \lim_{n \to \infty} \bigl(A_{n(z+\mu)}^{(n)}+B_{n(z+\mu)}^{(n)}\bigr) > 
 \lim_{n \to \infty} C_{n(z+\mu)}^{(n)}. 
\end{align}
Since the coefficients $A_x^{(n)}$, $B_x^{(n)}$ and $C_x^{(n)}$ are 
continuous functions of $x$, for any $\ep>0$ small enough, we have
\begin{align}
 \min_x A_x^{(n)} + \min_x B_x^{(n)}  > \max_x C_x^{(n)},
\end{align}
where $\max_x$ and $\min_x$ are taken in the range 
\begin{align} 
 x \in [n(z+\mu-\ep), n(z+\mu+\ep)].
\end{align}
Hereafter until \eqref{eq:AE:39}, we use the same range for $\max_x$ and $\min_x$.

The definition \eqref{eq:AE:G_n} yields
$G_n(z+\ep)-G_n(z)=\sum_{i=1}^{n \ep} p_n(n(z+\mu)+i)$.
Then, for $-\mu<z<0$, we have 
\begin{align}
\begin{split}
 & \left(\min_x A_x^{(n)}\right)\bigl(G_n(z+\ep)-G_n(z)\bigr) +
   \left(\min_x B_x^{(n)}\right)\bigl(G_n(z+\ep)-G_n(z)\bigr) \\
 &+\left(\min_x A_x^{(n)}\right)\bigl(p_n(n(z+\mu)) + p_n(n(z+\mu)-1) \\
 &\hspace{8em}                       -p_n(n(z+\mu+\ep)) - p_n(n(z+\mu+\ep)-1)\bigr) \\
 &+\left(\min_x B_x^{(n)}\right)\bigl(p_n(n(z+\mu))-p_n(n(z+\mu+\ep))\bigr) \\
=& \left(\min_x A_x^{(n)}\right)\bigl(G_n(z+\ep-\tfrac{2}{n})-G_n(z-\tfrac{2}{n})\bigr) +
   \left(\min_x B_x^{(n)}\right)\left(G_n(z+\ep-\tfrac{1}{n})-G_n(z-\tfrac{1}{n})\right) \\
=& \sum_{i=1}^{n \ep} \left(\min_x A_x^{(n)}\right) p_n(n(z+\mu)+i-2)
                     +\left(\min_x B_x^{(n)}\right) p_n(n(z+\mu)+i-1) \\
\le
 & \sum_{i=1}^{n \ep} A_{n(z+\mu)+i}^{(n)} \, p_n(n(z+\mu)+i-2) 
                     +B_{n(z+\mu)+i}^{(n)} \, p_n(n(z+\mu)+i-1) \\
=& \sum_{i=1}^{n \ep} C_{n(z+\mu)+i}^{(n)} \, p_n(n(z+\mu)+i)   
 \label{eq:AE:=rec} \\
\le
 & \sum_{i=1}^{n \ep} \left(\max_x C_x^{(n)}\right) p_n(n(z+\mu)+i) 
=  \left(\max_x C_x^{(n)} \right) \bigl(G_n(z+\ep)-G_n(z)\bigr).
\end{split}
\end{align}
Thus, we have
\begin{align}\label{eq:AE:39}
\begin{split}
&\left(\min_x A_x^{(n)}+\min_x B_x^{(n)}-\max_x C_x^{(n)}\right) 
 \left(G_n(z+\ep-\tfrac{1}{n})-G_n(z-\tfrac{1}{n})\right) \\
&\le -\left(\min_x A_x^{(n)}\right) \bigl(p_n(n(z+\mu)) + p_n(n(z+\mu)-1) \\
&\hspace{8em}
     -p_n(n(z+\mu+\ep)) - p_n(n(z+\mu+\ep)-1) \bigr) \\
&\phantom{M}
 -\left(\min_x B_x^{(n)}\right) \bigl(p_n(n(z+\mu))-p_n(n(z+\mu+\ep))\bigr). 
\end{split}
\end{align}
Proposition \ref{prp:AE:lim-p} guarantees that 
the right hand side goes to zero as $n \to \infty$.
Since the case $z=-\mu$ is not contained, the limit of 
$\left(\min_x A_x^{(n)}+\min_x B_x^{(n)}-\max_x C_x^{(n)}\right)$ is strictly positive.
Then, since $G_n(z+\ep)-G_n(z)\ge 0$, we have
\begin{align}
 G(z+\ep)-G(z) = \lim_{n \to \infty} \bigl(G_n(z+\ep)-G_n(z)\bigr) = 0.
\end{align}
Since $\lim_{n \to \infty}p(0)=0$ by Proposition \ref{prp:AE:lim-p} for example, 
the same claim holds for $-\mu \le z <0$.

In the case $0<z \le \mu$, the same argument works by reversing inequalities.
Therefore, we have proved that 
the limit distribution has a positive probability only at $Z=0$
in the case $\xi \le 1/2$.
By the system symmetry \eqref{eq:intro:sym}, we have the same statement
in the case $\xi \ge 1/2$.
\end{proof}

\section{Proof of central limit theorem, Theorem \ref{thm:intro:CLT}}\label{ss:CLT-prf}


We give a proof of the central limit Theorem \ref{thm:intro:CLT}.
Before starting the discussion, let us recall the ad-hoc derivation of 
the normal distribution $\Psi(x)$ in \eqref{eq:II:Psi}. 
We started from the recurrence relation of the pmf $p(x)$ in Lemma \ref{lem:intro:rec}, 
regarding it as the difference equation characterizing $p(x)$.
Then, assuming that there exists a differentiable limit distribution of $p(x)$ 
for $n \to \infty$, we derived a differential equation 
\eqref{eq:II:ode} or \eqref{eq:II:R} of the limit distribution 
from the difference equation of $p(x)$. 
The obtained differential equation was the one characterizing 
the normal distribution $\Psi(x)$.

Now, let us recall the statement.
We consider Type II limit with the assumption \eqref{eq:intro:II:asmp}, 
i.e., the limit $n \to \infty$ with fixed 
$\alpha=\frac{l}{n}$, $\beta =\frac{m-l}{n}$, $\gamma=\frac{k-l}{n}$ 
and $\delta=\frac{n-m-k+l}{n}$ satisfying \eqref{eq:intro:II:asmp}.
Then, for any real $t<u$, we have 
\begin{align}
 \lim_{n \to \infty} \Prb_n\Bigl[t \le \frac{X-n \mu}{\sqrt{n} \sigma} \le u\Bigr]
 = \frac{1}{\sqrt{2\pi}} \int_t^u e^{-s^2/2} \, d s, 
\end{align}
where $\Prb_n$ denotes the probability for $P_{n,n\xi,n\kappa,n\alpha}$,
$\xi=\alpha+\beta=\frac{m}{n}$, $\kappa=\alpha+\gamma=\frac{k}{n}$, 
and $\mu$, $\sigma$ have the following forms
\begin{align}
    \mu = \frac{1-\sqrt{D}}{2}, \quad
 \sigma = \sqrt{\frac{(\alpha+\gamma)\beta\delta}{D}}, \quad 
      D = 1-4(\alpha\gamma+\alpha\delta+\beta\gamma). 
\end{align}
%

As mentioned in Remark \ref{rmk:msD:sym}, 
the statement is preserved under the system symmetry \eqref{eq:intro:sym},
so we assume $m \le n-m$ in the following discussion.
Under this assumption, we can use the recurrence relation of the pmf $p(x)$, 
given in Lemma \ref{lem:intro:rec}, which will be the main tool of our proof.
In fact, the strategy is to make a sufficient refinement of the ad-hoc
derivation of the limit normal distribution $\Psi(x)$ recalled in the beginning.

Hereafter the word ``limit" means Type II limit, 
and we always assume the condition \eqref{eq:intro:II:asmp}.
We also use the fixed ratios $\xi=m/n$ and $\kappa=k/n$ as in Section \ref{ss:AE-prf}.
Then, by the condition \eqref{eq:intro:II:asmp}, 
we always have 
\begin{align}\label{eq:CLT:cond}
 \mu < \xi \wedge \kappa \wedge (1-\kappa), \quad \xi \le 1/2.
\end{align}

\subsection{Random variable $Y$ with non-compact support and estimation of tails}
\label{sss:CLT:Y}

In Section \ref{ss:AE-prf}, we used the random variable $Z=\frac{X-n \mu}{n}$ to 
show the law of large numbers for $X/n$ (Theorem \ref{thm:intro:E}).
For the central limit theorem, we consider instead the random variable
\begin{align}
 Y \ceq \frac{X - n \mu}{\sqrt{n} \sigma}, 
\end{align}
which has a non-compact support in the limit $n \to \infty$.
We denote the cdf for $Y$ by 
\begin{align}
 F_n(y) \ceq \sum_{u=0}^{n \mu+\sqrt{n} \sigma y} p(u \mid n,m,k,l).
\end{align}
Now, let us recall:

\begin{fct}[{Helly lemma \cite[2c.4 (i), p.117]{R}}]\label{fct:Helly}
There is a subsequence $\{n_i\}_i \in \bbN$ such that $F_{n_i}$ converges to 
a function $F$ at all continuity points of $F$. 
Moreover, $F$ is continuous from the left, bounded, and non-decreasing.
\end{fct}

Note that $F$ might not be the cdf of a probability distribution a priori.

In order to concentrate the discussion 
on the nearby points around $\frac{X}{n}=\mu$, 
we need to estimate the tails for $Y$, which is the purpose of this part.
The result is:

\begin{prp}\label{prp:tail}
Denote by $\Prb_n$ the probability for $P_{n,n \xi,n \kappa,n \alpha}$.
Then, we have
\begin{align}
 \lim_{R \to \infty} \limsup_{n \to \infty} \Prb_n[Y> R] = 0, \quad 
 \lim_{R \to \infty} \limsup_{n \to \infty} \Prb_n[Y<-R] = 0.
\end{align}
\end{prp}

Our proof uses the rewritten form \eqref{eq:AE:AP2} of the recursion for 
$p_n(x) \ceq p(x \midd n,m,k,l)$:
\begin{align}\label{eq:CLT:rec}
 p_n(x+1) - \zeta_{1,x} \, p_n(x) = 
 \zeta_{2,x} \bigl(p_n(x)- \zeta_{1,x} \, p_n(x-1)\bigr).
\end{align}
The relevant quantities are given as follows (see also \eqref{eq:AE:zeta}). 
\begin{align}
\begin{split}
  \eta_{1,x} &\ceq  -\frac{A_x^{(n)}}{C_x^{(n)}} = \eqref{eq:AE:eta1}, \quad 
  \eta_{2,x}  \ceq   \frac{A_x^{(n)}+B_x^{(n)}-C_x^{(n)}}{C_x^{(n)}}, \\
 \zeta_{1,x} &\ceq   \frac{\theta_x-\sqrt{\theta_x^2-4\eta_{1,x}}}{2}, \quad 
 \zeta_{2,x}  \ceq   \frac{\theta_x+\sqrt{\theta_x^2-4\eta_{1,x}}}{2}, \quad 
 \theta_x     \ceq 1+\eta_{1,x}+\eta_{2,x}.
\end{split}
\end{align}
In the limit $n \to \infty$, we have \eqref{eq:AE:eta-lim}: 
\begin{align}\label{eq:CLT:eta-lim}
 \lim_{n \to \infty} \eta_{1,n t}
 =\frac{(1- t)(t-\kappa)}{ t(1-\kappa-t)}, \quad 
 \lim_{n \to \infty} \eta_{2,n t}
 =\frac{(t-\mu)(t-\nu)(1-2t)^2}{(\xi-t)(1-\xi-t)t(1-\kappa-t)}.
\end{align}

\begin{proof}[{Proof of Proposition \ref{prp:tail}}]
Let us discuss the former probability $\Prb_n[Y>R]$ in the statement.
Let $R>0$ and set $R_1 \ceq R \sigma$.
Using $\ep>0$, we divide the probability into two parts as 
\begin{align}\label{eq:tail:LPR}
\begin{split}
&\Prb_n[Y>R] = \Prb_n\bigl[ X > n \mu+\sqrt{n}R_1 \bigr] \\
&            = \Prb_n\bigl[ X > n(\mu+\ep)+\sqrt{n}R_1 \bigr]
             + \Prb_n\bigl[ n(\mu +\ep)+\sqrt{n}R_1 \ge X > n \mu+\sqrt{n}R_1\bigr].
\end{split}
\end{align}
By the law of large numbers for $X/n$ (Theorem \ref{thm:LLN}), we have 
$\lim_{n \to \infty} \Prb_n\bigl[ X > n(\mu+\ep)+\sqrt{n}R_1 \bigr] = 0$.
So it is enough to treat the remaining probability.

For brevity, we denote $x_{1,n} \ceq n \mu+\sqrt{n}R_1$ and 
$x_{2,n} \ceq n(\mu +\ep)+\sqrt{n}R_1$.
Recalling the recursion \eqref{eq:CLT:rec}, we set 
\begin{align}\label{eq:tail:zeta2max}
 \zeta_{2,\max} \ceq \max_{x \in [x_{1,n},x_{2,n}]} \zeta_{2,x} \in (0,1).
\end{align}
The inequalities $0<\zeta_{2,\max}<1$ follows from Lemma \ref{lem:AE:lem1} (3), 
since we have $0 < \frac{x}{n}-\mu < 1-\kappa-\mu$ in the range concerned
(recall \eqref{eq:CLT:cond}).
The same argument shows $\zeta_{1,x}<0$ for $x \in [x_{1,n},x_{2,n}]$.

If $\ep>0$ is small enough, then $\zeta_{1,x}$ is either 
monotonically increasing or decreasing in the range $x \in [x_{1,n},x_{2,n}]$.
Hereafter until \eqref{eq:tail:dec-end}, 
we assume that $\abs{\zeta_{1,x}}$ is decreasing. It yields
\begin{align}
 p_n(x+1+i)+\abs{\zeta_{1,x+i}} p_n(x+i) \le 
 p_n(x+1+i)+\abs{\zeta_{1,x+i-1}} p_n(x+i).
\end{align}
Using the recursion \eqref{eq:CLT:rec} iteratively, we have
\begin{align}
 p_n(x+1+i)+\abs{\zeta_{1,x+i}} p_n(x+i) \le 
 \zeta_{2,\max}^i \bigl(p_n(x)+\abs{\zeta_{1,x}} p_n(x-1)\bigr).
\end{align}
By summation, we have
\begin{align}\label{eq:tail:AP8}
\begin{split}
&\sum_{j=0}^{\sqrt{n}} \bigl(p_n(x+1+i+j)+\abs{\zeta_{1,x+i+j}} p_n(x+i+j)\bigr) \\
&\le \zeta_{2,\max}^i 
 \sum_{j=0}^{\sqrt{n}} \bigl(p_n(x+j)+\abs{\zeta_{1,x+j}} p_n(x-1+j)\bigr).
\end{split}
\end{align}
We will use this inequality to estimate the remaining part 
$\Prb_n\bigl[ x_{2,n} \ge X > x_{1,n} \bigr] = \sum_{x=x_{1,n}+1}^{x_{2,n}} p_n(x)$.
The idea is to consider the range $[x_{1,n},x_{2,n}]$ modulo $\sqrt{n}$.
Replacing $i$ in \eqref{eq:tail:AP8} by $\ep i'$, we have 
\begin{align}\label{eq:tail:AP9}
\begin{split}
&\Prb_n\bigl[ x_{2,n} \ge X > x_{1,n} \bigr] = \sum_{x=x_{1,n}+1}^{x_{2,n}} p_n(x) 
 \le \sum_{i'=0}^{\ep \sqrt{n}} \sum_{j=0}^{\sqrt{n}} p_n(x_{1,n}+1+\sqrt{n}i'+j) \\
&\le \sum_{i'=0}^{\ep \sqrt{n}} \sum_{j=0}^{\sqrt{n}} \bigl(p_n(x_{1,n}+1+\sqrt{n}i'+j)
    +\abs{\zeta_{1,x_{1,n}+\sqrt{n}i'+j}} p_n(x_{1,n}+\sqrt{n}i'+j)\bigr) \\
&\le \sum_{i'=0}^{\ep \sqrt{n}} \zeta_{2,\max}^{\sqrt{n}i'} 
     \sum_{j =0}^{\sqrt{n}} \bigl(p_n(x_{1,n}+j)
    +\abs{\zeta_{1,x_{1,n}+j}} p_n(x_{1,n}-1+j) \bigr) \\
&\le (1-\zeta_{2,\max}^{\sqrt{n}})^{-1}
     \sum_{j=0}^{\sqrt{n}}\bigl(p_n(x_{1,n}+j)
    +\abs{\zeta_{1,x_{1,n}}} p_n(x_{1,n}-1+j)\bigr). 
\end{split}
\end{align}
In the last inequality we used $0<\zeta_{2,\max}<1$ in \eqref{eq:tail:zeta2max}.

Now, we want to take the limit $n \to \infty$ of the inequality \eqref{eq:tail:AP9}.
Set $c_1 \ceq \lim_{n \to \infty} \abs{\zeta_{1,x_{1,n}}}$ which surely exists.
As for the limit of $\zeta_{2,\max}$, let us observe that, since $\frac{x}{n} \to \mu$, 
$\eta_{2,x}$ is small by the expression of $\ol{\eta}_2$ in \eqref{eq:CLT:eta-lim}.
Since $\eta_{1,x} \sim \eta_{1,\mu}<0$ in the limit, we have 
\begin{align}\label{eq:tail:zeta_2x}
 \zeta_{2,x} = \frac{1+\eta_{1,x}+\eta_{2,x}+ 
                     \sqrt{(1-\eta_{1,x})^2+\eta_{2,x}^2+2(1+\eta_{1,x})\eta_{2,x}}}{2}
 \sim 1+\frac{\eta_{2,x}}{1-\eta_{1,x}}. 
\end{align}
A direct calculation using \eqref{eq:tail:zeta_2x} yields that 
$\zeta_{2,\max}=\zeta_{2,x_{1,n}}$ and 
\begin{align}\label{eq:tail:c2}
 c_2 \ceq \lim_{n \to \infty} \sqrt{n}(1-\zeta_{2,\max})
      = \frac{R_1 (\nu-\mu)(1-2\mu)}{\kappa(\xi-\mu)(1-\xi-\mu)} > 0.
\end{align}
Then, in \eqref{eq:tail:AP9}, we have 
$\lim_{n \to \infty} \zeta_{2,\max}^{\sqrt{n}}
=\lim_{n \to \infty} (1-\frac{c_2}{\sqrt{n}})^{\sqrt{n}}=e^{-c_2}$ and 
\begin{align}\label{eq:CLT:5522}
\begin{split}
&\limsup_{n \to \infty} \sum_{j=0}^{\sqrt{n}} 
 \bigl(p_n(x_{1,n}+j)+\abs{\zeta_{1,x_{1,n}}} p_n(x_{1,n}-1+j)\bigr)  \\
&\le(1+c_1) \bigl(F(\tfrac{R_1+1}{\sigma}+0)-F(\tfrac{R_1}{\sigma}-0)\bigr) 
 =  (1+c_1) \bigl(F(R+\tfrac{1}{\sigma}+0)-F(R-0)\bigr).
\end{split}\end{align}
Thus, the limit of the inequality \eqref{eq:tail:AP9} yields
\begin{align}
 \limsup_{n \to \infty} \Prb_n \bigl[x_{2,n} \ge X > x_{1,n} \bigr]
 \le (1-e^{-c_0})^{-1} (1+ c_{0,1}) \bigl(F(R+\tfrac{1}{\sigma}+0)-F(R-0)\bigr).
\end{align}
Note that $c_2>0$ by \eqref{eq:tail:c2}, 
so that the term $(1-e^{-c_0})^{-1}$ makes sense.
Since $F$ is the limit function obtained by Helly lemma (Fact \ref{fct:Helly}), 
we have $\lim_{R \to \infty} \bigl(F(R+\tfrac{1}{\sigma}+0)-F(R-0)\bigr) = 0$.
Hence, we obtain
\begin{align}\label{eq:tail:dec-end}
 \limsup_{n \to \infty} \Prb_n \bigl[x_{2,n} \ge X > x_{1,n} \bigr] = 0.
\end{align}

We can similarly treat the case when $|\zeta_{1,x}|$ is increasing 
in the range $x \in [x_{1,n} \sqrt{n}, x_{2,n}]$.
In fact, dividing the recursion \eqref{eq:CLT:rec} by $\abs{\zeta_{1,x}}$, we have 
\begin{align}
 \abs{\zeta_{1,x}}^{-1}p_n(x+1)+ p_n(x)
=\zeta_{2,x} \bigl(\abs{\zeta_{1,x}}^{-1} p_n(x)+ p_n(x-1)\bigr).
\end{align}
Since $\abs{\zeta_{1,x}}^{-1}$ is decreasing in this case, 
a similar argument as before works, and yields \eqref{eq:tail:dec-end}.

Therefore, we have \eqref{eq:tail:dec-end} in either case,
and combining it with \eqref{eq:tail:LPR}, we obtain the desired equality 
$\lim_{R\to \infty} \limsup_{n \to \infty} \Pr_n[Y>R] = 0$.

Finally, we treat the latter $\Pr_n[Y<-R]$ in the statement.
For $x \in [n(\mu-\ep),n \mu-R_1 \sqrt{n}]$, 
we consider the rewritten form of the recursion \eqref{eq:CLT:rec}:
\begin{align}
 p_n(x)+ \abs{\zeta_{1,x}} p_n(x-1) = 
 \zeta_{2,x}^{-1} \bigl(p_n(x+1)+ \abs{\zeta_{1,x}} p_n(x)\bigr). 
\end{align}
{}From \eqref{eq:CLT:eta-lim} and $\frac{x}{n}-\mu<0$, 
we can find that $\eta_{2,x}$ is positive and of order $O(1/\sqrt{n})$.
We also find that $1-\max_{x} \zeta_{2,x}^{-1} = O(1/\sqrt{n})$,
a similar argument as before works, and it yields the desired consequence
$\lim_{R\to \infty} \limsup_{n \to \infty} \Pr_n[Y<-R] = 0$.
\end{proof}

\subsection{The limit distribution}

Going back to the beginning of \S \ref{sss:CLT:Y},
we consider the random variable $Y \ceq \frac{X - n \mu}{\sqrt{n}\sigma}$, 
the cdf $F_n(y)$ and the limit function $F$.
By Proposition \ref{prp:tail}, we can focus on the interval 
\begin{align}
 n \mu-\sqrt{n}\sigma R \le X \le n \mu+\sqrt{n}\sigma R \iff 
 -R \le Y \le R.
\end{align}
We continue to use the recursion of $p_n(x)=p(x \midd n,m,k,l)$ as the main tool.
Let us write it again: 
\begin{align}\label{eq:CLT:AP3}
 \bigl(p_n(x+1)-p_n(x)\bigr) - \eta_{1,x} \bigl(p_n(x)-p_n(x-1)\bigr) = 
 \eta_{2,x} p_n(x).
\end{align}

Let us study the behavior of $p_n(x)$ around $x=n \mu$. 

\begin{lem}\label{lem:CLT:lem3}
If $\{x_n\}_n$ is a sequence in $I \ceq [n \mu-\sqrt{n}\sigma R, n \mu+\sqrt{n}\sigma R]$
such that $\{\frac{x_n-n\mu}{\sqrt{n}}\}_n$ converges, then we have 
$\lim_{n \to \infty}p_n(x_n)=0$.
\end{lem}

\begin{proof}
We assume that $p_n(x_n)$ converges to a positive value,
and prove the statement by contradiction.
We can also assume that $p_n(x_n-1)$ converges
(if not, then take a converging subsequence $\{n_j\}_j$).
Then, for every $i \in \bbZ_{\ge -1}$, we have the limit 
$a_i \ceq \lim_{n \to \infty} p_n(x_n+i)$.
In particular, we have $a_0>0$.
Since $\lim_{n \to \infty} \eta_{2,x}=0$ for $x$ in the range concerned,
the recursion \eqref{eq:CLT:AP3} yields 
\begin{align}\label{eq:CLT:lrec}
 a_{i+1}-a_i = -\lambda(a_i-a_{i-1}),
\end{align}
with $\lambda \ceq -\lim_{n \to \infty}\eta_{1,x}>0$.
On the other hand, we have 
$\sum_{i=0}^ N a_i \le \liminf_{n \to \infty} \sum_{x \in I} p_n(x)$ 
for any $N \in \bbN$, which yields 
\begin{align}
 \sum_{i=0}^{\infty} a_i \le \liminf_{n \to \infty} \sum_{x \in I} p_n(x).
\end{align}
If $a_0=a_{-1}$, then $a_i=a_0>0$ for any $i$, and 
the summation $\sum_{i=0}^{\infty} a_i$ diverges, which is a contradiction.
Thus, we can assume $a_0 \neq a_{-1}$.
If $\lambda \ge 1$, then \eqref{eq:CLT:lrec} implies that 
$\sum_{i=0}^{\infty} a_i$ diverges, which is a contradiction.
If $\lambda<1$, then $a_i$ converges to a point $a$
in the open interval $(a_0,a_{-1})$ or $(a_{-1},a_0)$.
Then, $a>0$, and $\sum_{i=0}^{\infty} a_i$ diverges, which is a contradiction.
\end{proof}

Hereafter we fix $y \in \bbR$ and set 
\begin{align}
 x_0 \ceq n \mu+\sqrt{n}\sigma y, \quad x_{\pm} \ceq x_0 \pm \sqrt{n}{\sigma} \ep.
\end{align}
We only treat the case $x \in [x_0,x_+]$,
when $\eta_{2,x}=O(1/\sqrt{n})$ and has the same sign as $y$.

For $x \in [x_0,x_+]$, the function $\eta_{1,x}$ of $x$ is 
either monotonically increasing or decreasing.
We assume it is increasing, and set 
\begin{align}\label{eq:CLT:eta1od}
 \eta_1^o \ceq \eta_{1,x_0}, \quad d(x) \ceq \eta_{1,x}-\eta_1^o.
\end{align}
We can see from \eqref{eq:CLT:eta-lim} that $d(x)=O(\ep/ \sqrt{n})$ and $d(x) \ge 0$.
The decreasing case can be treated similarly by setting $\eta_{1}^o \ceq \eta_{1,x_+}$,
so we omit the detail.

In the limit $n \to \infty$, the quantity $\eta_1^o$ converges to 
a constant independent of $y$ and $\ep$. We denote it by 
\begin{align}\label{eq:CLT:beta1o}
 \ol{\eta}_1^o \ceq \lim_{n \to \infty} \eta_1^o.
\end{align}
Using these quantities, we have:

\begin{lem}\label{lem:CLT:lem4}
Define $c_1^{\pm}$ and $c_2$ by 
\begin{align}
 c_1^+ \ceq \lim_{n \to \infty} \sqrt{n}\sigma \left(\max_{x'} \eta_{2,x'}\right), \ 
 c_1^- \ceq \lim_{n \to \infty} \sqrt{n}\sigma \left(\min_{x'} \eta_{2,x'}\right), \ 
 c_2   \ceq \lim_{n \to \infty} \sqrt{n}\sigma \left(\max_{x'} d(x')\right),
\end{align}
where $\min_{x'}, \max_{x'}$ are taken in the interval $[x_-, x_+]$.
Then, for any $0<\ve'<\ve$, we have 
\begin{align}\label{eq:CLT:lem4}
\begin{split}
&(c_1^- -c_2) \int_0^{\ep'} \bigl(F(y+\ep-t) -F(y-t)\bigr) \, d t \\
&\qquad \le 
 (1-\ol{\eta}_1^o) \bigl(F(y+\ep)- F(y+\ep-\ep')- F(y) +F(y-\ep') \bigr) \\
&\qquad \qquad \le 
 (c_1^+ +c_2) \int_0^{\ep'} \bigl(F(y+\ep-t) -F(y-t)\bigr) \, d t. 
\end{split}
\end{align}
\end{lem}

\begin{proof}
The recursion \eqref{eq:CLT:AP3} can be written as
\begin{align}
 \bigl(p_n(x+1)-p_n(x)\bigr)-\eta_{1}^o \bigl(p_n(x)-p_n(x-1)\bigr) 
=\eta_{2,x} p_n(x) + d(x) \bigl(p_n(x)-p_n(x-1)\bigr).
\end{align}
It yields the inequalities
\begin{align}\label{eq:CLT:AP4}
\begin{split}
&\left(\min_{x'} \eta_{2,x'}\right) p_n(x) - \left(\max_{x'} d(x)\right) p_n(x-1) \\
&\le \bigl(p_n(x+1)-p_n(x)\bigr) -\eta_{1}^o \bigl(p_n(x)-p_n(x-1)\bigr) 
 \le \max_{x'}\bigl(\eta_{2,x'}+d(x')\bigr) p_n(x),
\end{split}
\end{align}
where $\min_{x'}, \max_{x'}$ are taken in the interval $[x_-, x_+]$.
Applying the second inequality of \eqref{eq:CLT:AP4} on the summation 
$p_n(x_+)-p_n(x_0)=\sum_{i=1}^{\sqrt{n}\sigma\ep} \bigl(p_n(x_0+i)-p_n(x_0+i-1) \bigr)$,
we have 
\begin{align}
\begin{split}
&p_n(x_+)-p_n(x_0) -\eta_1^o \bigl(p_n(x_+ -1)-p_n(x_0 -1) \bigr) \\
&\le \max_{x'}\bigl(\eta_{2,x'}+d(x')\bigr) 
 \sum_{i=1}^{\sqrt{n}\sigma\ep} p_n(x_0+i-1) \\
&\quad 
 =\max_{x'}\bigl(\eta_{2,x'}+d(x')\bigr) 
  \bigl(F_n(y+\ep-\tfrac{1}{\sqrt{n}\sigma})-F_n(y) \bigr).
\end{split}
\end{align}
Take any $0<\ep' <\ep$ as in the statement.
Then, for $j \in [0, \sqrt{n}\sigma\ep']$, the same argument yields
\begin{align}
\begin{split}
 p_n(x_+ -j)-p_n(x_0-j) &- \eta_{1}^o \bigl(p_n(x_+ -j-1)-p_n(x_0-j-1) \bigr) \\
&\le \max_{x'} \bigl(\eta_{2,x'}+d(x')\bigr) 
 \bigl(F_n(y+\ep-\tfrac{1}{\sqrt{n}\sigma}-j)-F_n(y-j) \bigr).
\end{split}
\end{align}
Summation over $j$ yields 
\begin{align}\label{eq:CLT:Fineq}
\begin{split}
&F_n(y+\ep)-F_n(y+\ep-\ep')-F_n(y)+F_n(y-\ep) \\
&-\eta_1^o \Bigl(
 F_n(y+\ep-\tfrac{1}{\sqrt{n}\sigma})-F_n(x+\ep-\ep'-\tfrac{1}{\sqrt{n}\sigma}) 
-F_n(y    -\tfrac{1}{\sqrt{n}\sigma})+F_n(y    -\ep'-\tfrac{1}{\sqrt{n}\sigma}) \Bigr) \\
&\le \max_{x'} \bigl(\eta_{2,x'}+d(x')\bigr) \sum_{j=0}^{\sqrt{n}\sigma\ep'}
 \bigl(F_n(y+\ep-\tfrac{j+1}{\sqrt{n}\sigma})
      -F_n(y-\tfrac{j}{\sqrt{n}\sigma})\bigr).
\end{split}
\end{align}
Let us admit the following for a while:
\begin{align}\label{eq:CLT:Fint}
 \lim_{n \to \infty} \frac{1}{\sqrt{n}\sigma} \sum_{j=0}^{\sqrt{n} \sigma \ep}
 F_n(y-\tfrac{j}{\sqrt{n}\sigma}) = \int_0^{\ep} F(y-t) \, d t.
\end{align}
Then, taking $n \to \infty$ of the inequality \eqref{eq:CLT:Fineq}, we obtain
\begin{align}
\begin{split}
 (1-\ol{\eta}_1^o) &\bigl(F(y+\ep)-F(y+\ep-\ep')-F(y)+F(y-\ep) \bigr) \\
&\le (c_1^++c_2)\int_0^{\ep'} \bigl(F(y+\ep-t) -F(y-t)\bigr) \, d t, 
\end{split}
\end{align}
where $\ol{\eta}_1^o$ is given in \eqref{eq:CLT:beta1o},
and  $c_1^+,c_2$ are given in the statement.
This is the second of the desired inequalities \eqref{eq:CLT:lem4}.


Now, we show \eqref{eq:CLT:Fint}.
Divide the interval $[0,\ep]$ by the points $\{\tfrac{j}{\sqrt{n} \sigma}\}_j$.
We also consider another coarser division by $\{\tfrac{i}{N}\}_i$ 
which has the first division as a refinement.
Since $F_n(y-t)$ is a decreasing function of $t$, we have
\begin{align}
\begin{split}
 \frac{1}{N} \sum_{i=0}^{N \ep} F_n(y-\tfrac{i+1}{N}) 
&\le 
 \frac{1}{\sqrt{n}\sigma} \sum_{j=0}^{\sqrt{n}\sigma \ep} 
 F_n(y-\tfrac{j+1}{\sqrt{n} \sigma}) \\
&\le 
 \frac{1}{\sqrt{n}\sigma} \sum_{j=0}^{\sqrt{n}\sigma \ep} 
 F_n(y-\tfrac{j}{\sqrt{n} \sigma}) \le 
 \frac{1}{N} \sum_{i=0}^{N \ep} F_n(y-\tfrac{i}{N}).
\end{split}
\end{align}
Now, fixing $N$, we take $n \to \infty$. 
Since it has been already shown that $F$ is continuous, we have 
\begin{align}\label{eq:CLT:FFn}
\begin{split}
 \frac{1}{N} \sum_{i=0}^{N \ep} F(y-\tfrac{i+1}{N}) 
&\le \liminf_{n \to \infty} \frac{1}{\sqrt{n}\sigma} \sum_{j=0}^{\sqrt{n}\sigma \ep}
 F_n(y-\tfrac{j}{\sqrt{n} \sigma}) \\
&\le \limsup_{n \to \infty} \frac{1}{\sqrt{n}\sigma} \sum_{j=0}^{\sqrt{n}\sigma \ep}
 F_n(y-\tfrac{j}{\sqrt{n} \sigma}) 
 \le \frac{1}{N} \sum_{i=0}^{N \ep} F(y-\tfrac{i}{N}).
\end{split}
\end{align}
Next we take $N \to \infty$. 
Then, the leftmost and rightmost sides of \eqref{eq:CLT:FFn} converge to 
$\int_0^{\ep} F(y-t) \, d t$, 
so that both middle limits are equal to the same value,
which is the desired \eqref{eq:CLT:Fint}.
Therefore, we finished the proof of the second inequality in  \eqref{eq:CLT:lem4}.

As for the first inequality in \eqref{eq:CLT:lem4},
since $\ol{\eta_1}=\lim_{n \to \infty} \eta_{1,x_+}$,
a quite similar argument works using the first inequality in \eqref{eq:CLT:AP4}.
We omit the detail.
\end{proof}

Now, we can show the central limit theorem, Theorem \ref{thm:intro:CLT}.

\begin{proof}[{Proof of Theorem \ref{thm:intro:CLT}}]
Take any $0< \ep' <\ep$, and 
consider the the inequalities \eqref{eq:CLT:lem4} in Lemma \ref{lem:CLT:lem3}.
Dividing them by $\ep'$ and taking the limit $\ep' \to 0$,
we find that the the derivative $f$ of $F$ exists and 
\begin{align}\label{eq:CLT:FfF}
\begin{split}
     (c_1^- -c_2)  \bigl(F(y+\ep-t)-F(y-t)\bigr) 
&\le (1-\ol{\eta}_1^o) \bigl(f(y+\ep)-f(y)\bigr) \\
&\le (c_1^+ +c_2)  \bigl(F(y+\ep-t)-F(y-t)\bigr).
\end{split}
\end{align}
As noted right after \eqref{eq:CLT:eta1od}, 
we have $d(x)=\eta_{1,x}-\eta_1^o = O(\ep/\sqrt{n})$,
which implies $\lim_{\ep \to 0}c_2=0$.
A direct calculation using \eqref{eq:CLT:eta-lim} shows 
\begin{align}
 \lim_{\ep \to 0} c_1^+ = \lim_{\ep \to 0} c_1^- =  y d_1, \quad
 d_1 \ceq \frac{\sigma^2 (\mu-\nu)}{(\xi-\mu)(1-\xi-\mu)} 
        \frac{(1-2\mu)^2}{\mu(1-\kappa-\mu)}.
\end{align}
Thus, dividing \eqref{eq:CLT:FfF} by $\ep$ and taking $\ep \to 0$, 
we find that $f$ is differentiable and $d_1 y f(y) = (1-\ol{\eta}_1^o) f'(y)$, i.e.,
\begin{align}
 f'(y) = \frac{d_1 y}{1-\ol{\eta}_1^o} f(y).
\end{align}
Finally, recalling the calculation in \eqref{eq:II:sigma^2}, we have 
\begin{align}
 \frac{d_1}{1-\ol{\eta}_1^o} = -\sigma^2
 \frac{(\nu-\mu)(1-2\mu)}{\kappa (\xi-\mu)(1-\xi-\mu)} = -1,
\end{align}
which means that $f$ satisfies 
the same differential equation of the standard normal distribution.
Then, by Proposition \ref{prp:tail}, 
we find that $f$ is a probability distribution function, i.e., 
having the standard normalization factor $1/\sqrt{2\pi}$.
The proof is now finished.
\end{proof}

\section{Asymptotic analysis beyond central limit theorem} 
\label{ss:II:Cas}


In this subsection, we give an asymptotic analysis 
beyond the central limit theorem, Theorem \ref{thm:intro:CLT}.
We consider the Type II limit with the condition \eqref{eq:intro:II:asmp}.
Then, $\eta<1/4$, $D \ceq 1-4\eta>0$ and the following theorem holds.

\begin{thm}\label{thm:EV}
In Type II limit with the condition \eqref{eq:intro:II:asmp}, we have 
\begin{align}\label{eq:EV:HH3}
 \lim_{n \to \infty} \bbV \Bigl[\frac{X}{\sqrt{n}}\Bigr] = 
 \lim_{n \to \infty} \bbE \Bigl[\Big(\frac{X-n\mu}{\sqrt{n}}\Big)^2\Bigr] = 
 \sigma^2 \ceq \frac{(\alpha+\gamma)\beta \delta}{D}.
\end{align}
with $\mu\ceq(1-\sqrt{D})/2$. Moreover, we have 
\begin{align}\label{eq:EV:HH2} 
 \bbE[X] = n \mu + \phi + o(n^{0}), \quad 
 \phi \ceq \frac{\sigma^2-\mu}{1-2\mu}. 
\end{align}
\end{thm}

The main tool of the proof is Theorem \ref{thm:Cas}, where we derived the relation of 
the expectation $\bbE[X]$ and the variance $\bbV[X]$:
\begin{align}\label{eq:EV:201}
 \frac{\bbV[X]}{n^2} = 
 \frac{\bbE[X]}{n}\left(1-\frac{\bbE[X]}{n}\right)+\frac{\bbE[X]}{n^2} - \eta, \quad 
 \eta \ceq \alpha \gamma+\alpha\delta+\beta\gamma.
\end{align}
Further, we have
\begin{align}\label{eq:EV:20T}
 \lim_{n \to \infty} \bbE \Bigl[\Big(\frac{X-n\mu}{\sqrt{n}}\Big)^j\Bigr] 
&<\infty 
\end{align}
for $j=3,4, \ldots.$
The proof consists of two steps.
In the first step, we show the first equality in \eqref{eq:EV:HH3} 
and the estimation \eqref{eq:EV:HH2} using \eqref{eq:EV:201} under a certain assumption.
In the second step, we show the assumption.

\begin{proof}[{Step 1 of Proof of Theorem \ref{thm:EV}}]
Assuming
\begin{align}\label{eq:EV:MK}
 \lim_{n \to \infty} \bbE \Bigl[\frac{X-n\mu}{\sqrt{n}}\Bigr] = 0, \quad 
 \lim_{n \to \infty} \bbE \Bigl[\Bigl(\frac{X-n\mu}{\sqrt{n}}\Bigr)^2\Bigr] = \sigma^2,
\end{align}
we show the desired statements.
First, note that the estimations in \eqref{eq:EV:MK} yield 
\begin{align}\label{eq:EV:MKT}
 \lim_{n \to \infty} \tfrac{1}{n}\bbV[X] = \sigma^2.
\end{align} 
Second, the relation \eqref{eq:EV:201} implies 
the quadratic equation for $\bbE[X/n]$:
\begin{align}
 \bbE[\tfrac{X}{n}]^2 - 
 \bigl(1+\tfrac{1}{n}\bigr) \bbE[\tfrac{X}{n}] + \tfrac{1}{n^2}\bbV[X] + \eta = 0.
\end{align}
Since $\bbE[X]/n \le 1/2$,
we have
\begin{align}
 \bbE[\tfrac{X}{n}] = \tfrac{1}{2}\Bigl(1+\tfrac{1}{n}
  -\sqrt{1-4\eta+\tfrac{2}{n}+ \tfrac{1-4\bbV[X]}{n^2}} \Bigr) .
\end{align}
Since $\eta< 1/4$, we can expand the right hand side using \eqref{eq:EV:MKT} as
\begin{align}
 \bbE[\tfrac{X}{n}] = 
 \tfrac{1-\sqrt{1-4\eta}}{2}+\tfrac{1}{n}\tfrac{1}{\sqrt{1-4\eta}}
 \Bigl(\sigma^2- \tfrac{1-\sqrt{1-4\eta}}{2}\Bigr) + o(\tfrac{1}{n}) = 
 \mu + \tfrac{\phi}{n} + o(\tfrac{1}{n}).
\end{align}
\end{proof}

In the next step, we show the assumption \eqref{eq:EV:MK}.
Let us sketch the outline and the strategy.
A short discussion claims that it is enough to show 
the tail estimation in \eqref{eq:EV:BH2}.
The idea of its proof is to modify appropriately the arguments in the proof of 
the central limit Theorem \ref{thm:intro:CLT}, where we mainly used 
the recursion \eqref{eq:AE:NAP} of the pmf $p_n(x)=p(x \midd n,m,k,l)$:
\begin{align}\label{eq:EV:NAP0}
 A_x^{(n)} p_n(x-1) + B_x^{(n)} p_n(x) = C_x^{(n)} p_n(x+1).
\end{align}
See the lines below \eqref{eq:AE:NAP} for the precise expressions of 
the coefficients $A_x^{(n)}$, $B_x^{(n)}$ and $C_x^{(n)}$.
In the proof of Theorem \ref{thm:intro:CLT}, 
we used the random variable $Y=\frac{X-n \mu}{\sqrt{n}}$
and the cdf $F_n(y) \ceq \sum_{u=0}^{\sqrt{n} \sigma y+n \mu} p_n(u)$.
For the present Theorem \ref{thm:EV}, we will use the same variable $Y$,
but use the modified functions \eqref{eq:EV:olp} and \eqref{eq:EV:GF} 
for the evaluation.
Correspondingly we will use a modified recurrence relation \eqref{eq:EV:NAP2}.

\begin{proof}[{Step 2 of Proof of Theorem \ref{thm:EV}}]
The central limit Theorem \ref{thm:intro:CLT} yields
\begin{align}
 \lim_{R \to \infty} \lim_{n \to \infty} 
 \sum_{i=- \sqrt{n}R}^{\sqrt{n}R} \frac{i}{\sqrt{n}} \, p_n(n \mu +i) = 0, \quad
 \lim_{R \to \infty} \lim_{n \to \infty} 
 \sum_{i=- \sqrt{n}R}^{\sqrt{n}R} \frac{i^2}{n} \, p_n(n \mu +i) = \sigma^2.
\end{align}
Therefore, if we show the following relations, 
then the above ones yield \eqref{eq:EV:MK}.
\begin{align}
 \lim_{R \to \infty}\lim_{n \to \infty} 
 \sum_{i= \sqrt{n} R}^n \frac{|i|}{\sqrt{n}} 
 \bigl(p_n(n \mu-i)+p_n(n \mu +i)\bigr) &= 0, \label{eq:EV:BH1} \\
 \lim_{R \to \infty}\lim_{n \to \infty} 
 \sum_{i= \sqrt{n} R}^n \frac{i^2}{n} \bigl(p_n(n \mu -i)+p_n(n \mu +i)\bigr) &= 0. 
 \label{eq:EV:BH2}
\end{align}
Since \eqref{eq:EV:BH2} implies \eqref{eq:EV:BH1}, 
it is enough to show \eqref{eq:EV:BH2}.

For this aim, as mentioned before this step 2 of the proof, we define 
$\ol{p}_n(x)$, $\ol{A}_x^{(n)}$, $\ol{B}_x^{(n)}$ and $\ol{C}_x^{(n)}$ as
\begin{gather}\label{eq:EV:olp}
 \ol{p}_n (n \mu +i) \ceq \frac{i^2}{n} \, p_n(n \mu +i), \\
 \ol{A}_{n \mu +i}^{(n)} \ceq A_{n \mu +i}^{(n)} \frac{i^2}{(i-1)^2}, \quad 
 \ol{B}_{n \mu +i}^{(n)} \ceq B_{n \mu +i}^{(n)}, \quad 
 \ol{C}_{n \mu +i}^{(n)} \ceq C_{n \mu +i}^{(n)} \frac{i^2}{(i+1)^2}.
\end{gather}
Then, the recurrence relation \eqref{eq:EV:NAP0} is rewritten as
\begin{align}\label{eq:EV:NAP2}
 \ol{A}_x^{(n)} \ol{p}_n(x-1) + \ol{B}_x^{(n)} \ol{p}_n(x) = 
 \ol{C}_x^{(n)} \ol{p}_n(x+1) \quad (0<x<m).
\end{align}
We also define 
\begin{align}\label{eq:EV:GF}
 \ol{G}_n(z) \ceq \sum_{u=0}^{n(z+\mu)} \ol{p}_n(u), \quad
 \ol{F}_n(y) \ceq \sum_{u=0}^{\sqrt{n}\sigma y+n \mu} \ol{p}_n(u).
\end{align}
To show \eqref{eq:EV:BH2}, we need some evaluation of $\ol{F}_n(y)$.
To concentrate such an evaluation around the point $y=0$, we evaluate 
$\ol{G}_n(z)$ first, as we did to show the central limit Theorem \ref{thm:intro:CLT}.

Now, recall the proof of the law of large numbers (Theorem \ref{thm:LLN}).
Similarly as there, we take and fix $z<0$ until \eqref{eq:EV:5525}.
Then, for any small $\ep>0$, we have 
\begin{align}\label{eq:EV:AB>C}
 \lim_{n\to \infty} \left(\min_x \ol{A}_x^{(n)}+\min_x \ol{B}_x^{(n)}
                         -\max_x \ol{C}_x^{(n)}\right) > 0,
\end{align}
where $\max_x$ and $\min_x$ are taken in the range 
\begin{align} 
 x \in [n(z+\mu-\ep), n(z+\mu+\ep)].
\end{align}
Now, the same argument as \eqref{eq:AE:39} shows 
\begin{align}
 &\left(\min_x \ol{A}_x^{(n)}+\min_x \ol{B}_x^{(n)}-\max_x \ol{C}_x^{(n)}\right) 
  \left(\ol{G}_n(z+\ep-\tfrac{1}{n})-\ol{G}_n(z-\tfrac{1}{n})\right) 
 \nonumber \\
 &\le -\left(\min_x \ol{A}_x^{(n)}\right) 
       \bigl(\ol{p}_n(n(z+\mu))     + \ol{p}_n(n(z+\mu)-1) 
 \nonumber \\
 &\hspace{8em}
            -\ol{p}_n(n(z+\mu+\ep)) - \ol{p}_n(n(z+\mu+\ep)-1)\bigr) 
 \nonumber \\
 &\phantom{M}
 -\left(\min_x \ol{B}_x^{(n)}\right) 
  \bigl(\ol{p}_n(n(z+\mu)) - \ol{p}_n(n(z+\mu+\ep))\bigr).
 \label{eq:AE:39B}
\end{align}
Then, by \eqref{eq:EV:AB>C}, we have
\begin{align}\label{eq:EV:5525}
 \lim_{n \to \infty} \left(\ol{G}_n(z+\ep)-\ol{G}_n (z)\right) = 0.
\end{align}
The same statement holds for $z>0$. 
Thus, for any small $\ep>0$, we have
\begin{align}\label{eq:EV:NAO}
 \lim_{n \to \infty} \ol{G}_n(-\ep)=
 \lim_{n \to \infty} \left(\ol{G}_n(1)-\ol{G}_n (\ep)\right) = 0.
\end{align}

Next, similarly as in the proof of the central limit Theorem \ref{thm:intro:CLT},
we define $\ol{F}$ as the limit function of $F_n(y)-F_n(0)$.
Recalling the proof of Proposition \ref{prp:tail} of the tail estimation,
we take $R>0$ and set $x_{1,n} \ceq n \mu+\sqrt{n}R \sigma$ and 
$x_{2,n} \ceq n(\mu +\ep)+\sqrt{n}R \sigma$.
Then, in the same way as \eqref{eq:CLT:5522}, there exist $c_0,c_{0,1}>0$ such that
\begin{align}\label{eq:EV:MMK2}
 \limsup_{n \to \infty} \sum_{x=x_{1,n}}^{x_{2,n}} \ol{p}_n(x) 
 \le (1-e^{-c_0})^{-1} (1+ c_{0,1}) 
 \bigl(\ol{F}(R+\tfrac{1}{\sigma}+0)-\ol{F}(R-0)\bigr).
\end{align}
By \eqref{eq:EV:NAO} and \eqref{eq:EV:MMK2}, we have
\begin{align}
 \limsup_{n \to \infty} \sum_{x=n \mu+\sqrt{n}R_1}^{\infty} \ol{p}_n(x) 
 \le (1-e^{-c_0})^{-1} (1+ c_{0,1}) \bigl(\ol{F}(R+\tfrac{1}{\sigma}+0)-\ol{F}(R-0)\bigr).
\end{align}
Hence, the property of the limit function $F$ implies
\begin{align}\label{eq:EV:AJO1}
 \lim_{R \to \infty} \limsup_{n \to \infty} 
 \sum_{x=n \mu+\sqrt{n}R_1}^{\infty} \ol{p}_n(x) = 0.
\end{align}
In the same way, we can show 
\begin{align}\label{eq:EV:AJO2}
 \lim_{R \to \infty} \limsup_{n \to \infty} 
 \sum_{x=0}^{n \mu-\sqrt{n}R_1} \ol{p}_n(x)  =0.
\end{align}
The combination of \eqref{eq:EV:AJO1} and \eqref{eq:EV:AJO2} implies \eqref{eq:EV:BH2}.

Further, extending the above method, 
we can show the following for any positive integer $j$:
\begin{align}\label{MKB}
 \lim_{n \to \infty} \bbE \Bigl[\abs{\tfrac{X-n\mu}{\sqrt{n}}}^j \Bigr] = 
 \int_{-\infty}^{\infty} |y|^j F(d y) <\infty.
\end{align}
\end{proof}

\begin{rmk}
{}From the relation \eqref{eq:EV:201}, we can also derive
\begin{align}
 \lim_{n \to \infty} \bbV[X/n] = 0 \iff \lim_{n \to \infty} \bbE[X/n] = \mu
\end{align}
using a version of the compactness theorem.
It claims that the law of large numbers for $X/n$ holds 
if and only if the asymptotic expectation is equal to $\mu$.
Although this equivalence is unnecessary in the logic of this note,
it helped us to make sure that the quantity $\mu$ gives 
the correct asymptotic expectation in the early stages of the study.
\end{rmk}

\section{Further analysis in special case $\alpha \gamma=0$}\label{ss:II:sp}
\subsection{Proof of Proposition \ref{prp:5.6.2}}
The aim of this section is to show Proposition \ref{prp:5.6.2}.
For this aim,
to address a more detailed analysis than Section \ref{S2-2},
we consider Type II limit in the case $\alpha \gamma=0$,
corresponding to the case $K L=0$. 
Although \ref{ss:AE-prf} and \ref{ss:CLT-prf} derived the law of large numbers
and the central limit theorem, they did not derive the tight exponential evaluation of
the tail probability and the asymptotic behavior of the variance.
In the special case $K L=0$, 
we can evaluate the asymptotic behavior of the cumulant generating function.
Hence, we can discuss the tight exponential evaluation of the tail probability as follows.

Using the binary entropy $ h(t) \ceq -t \log t -(1-t)\log (1-t)$,
we define functions $f(t)$ and $u(s)$ as
\begin{align}
f(t) &\ceq h(t)-h(\alpha+\beta)+\alpha \, h\Bigl(\frac{t}{\alpha}\Bigr)
        -(\alpha+\beta) \, h\Bigl(\frac{t}{\alpha+\beta}\Bigr) 
        +(1-\alpha-t) \, h\Bigl(\frac{\beta}{1-\alpha-t}\Bigr),\\
 u(s) &\ceq \max_{t}\bigl(s t+f(t)\bigr).\label{eq:ac=0:u}
\end{align}

Then, as shown below, we have the following proposition.

\begin{prp}\label{prp:5.6.2B}
In Type II limit with $\alpha \beta \delta >0$, $\gamma=0$ 
and $0< \xi = \alpha+\beta \le 1/2$, we have
\begin{align}
 \lim_{n \to \infty} \tfrac{-1}{n} \log \Prb_n \bigl[\tfrac{X}{n} \ge R \bigr]
 &= \max_{s \ge 0} \bigl(s R -u(s)\bigr), \label{NBCV0} \\
 \lim_{n \to \infty} \tfrac{-1}{n} \log \Prb_n \bigl[\tfrac{X}{n} \le R \bigr]
 &= \max_{s \le 0} \bigl(s R -u(s)\bigr).\label{NBCV1}
\end{align} 
\end{prp}

As for the function $u(s)$ in \eqref{eq:ac=0:u}, 
we can compute the derivatives $u'(0)$ and $u''(0)$ as follows.
\begin{prp}\label{NMO}
We have 
\begin{align}\label{eq:ac=0:u1u2}
 u' (0) = \mu, \quad 
 u''(0) = \frac{-1}{f''(\mu)} = \frac{\alpha \beta(1-\xi)}{1-4\alpha(1-\xi)}
        = \frac{(\alpha+\gamma) \beta \delta}{D} = \sigma^2.
\end{align}
\end{prp}

Proposition \ref{NMO} can be shown as follows. 
The function $u(s)$ of \eqref{eq:ac=0:u} is 
the Legendre transform of the smooth convex function $-f(t)$.
Hence, for $s$ and $t$ related by $s=-f'(t)$, we have $t=u'(s)$ and $u''(s) f''(t)=-1$.
Since $0=f'(\mu)$ by \eqref{eq:ac=0:sol}, we obtain 
\eqref{eq:ac=0:u1u2}.

The combination of the above two propositions implies that
the probability $\Prb_n \bigl[| \tfrac{X}{n} - \mu |\ge \epsilon \bigr]$
goes to zero exponentially for any $\epsilon >0$
under the same assumption as Proposition \ref{prp:5.6.2B}.
Hence, we obtain Proposition \ref{prp:5.6.2}.

\subsection{Proof of Proposition \ref{prp:5.6.2B}}

Before taking the limit $n \to \infty$, 
the parameters concerned are $(n,m,k,l)$ satisfying $K \ceq k-l=0$ or $L \ceq l=0$.
Let us focus on the former case $k=l$, and also assume $m \le n-m$.
Then, the pmf $p(x)$ is computed in Proposition \ref{prp:k=l}:
\begin{align}\label{eq:ac=0:p}
 p(x) = \frac{n-2x+1}{n-x+1} \frac{\binom{n}{x}}{\binom{n}{m}}
        \frac{\binom{l}{x}}{\binom{m}{x}} \binom{n-l-x}{m-l} 
\end{align}
for $x \le k=l$.
Below we assume $l>0$, $M \ceq m-l>0$ and $N \ceq n-m-k+l>0$
because of \eqref{eq:intro:II:asmp}.
In this section, we consider the Type II limit 
with $\gamma=0$, i.e., the limit $n \to \infty$ fixed ratios 
\begin{align}
 \gamma=\tfrac{k-l}{n}=0, \ \alpha=\tfrac{l}{n}>0, \ 
 \beta=\tfrac{M}{n}>0, \ \delta=\tfrac{N}{n}>0, \ 
 \xi=\alpha+\beta=\tfrac{m}{n} \le 1/2.
\end{align} 
Now, we employ the quantities defined before, which have the following expressions
\begin{align}\label{eq:ac=0:munu}
\begin{split}
&\eta = \alpha \delta = \alpha(1-\xi), \quad D = 1-4 \alpha(1-\xi), \\
&\mu = \frac{1-\sqrt{1-4\alpha(1-\xi)}}{2}, \quad 
 \nu = \frac{1+\sqrt{1-4\alpha(1-\xi)}}{2}, \quad 
 \sigma = \sqrt{\frac{\alpha \beta \delta}{D}}.
\end{split}
\end{align}
Note that the assumption $0<\xi \le 1/2$ and $0<\alpha<1/2$ yields $D, \sigma >0$ and 
\begin{align}\label{eq:ac=0:mhn}
 \mu < \tfrac{1}{2} < \nu.
\end{align}

To prove Proposition \ref{prp:5.6.2},
below we discuss how to evaluate the tail probability in this case.
Using the binary entropy by
$ h(t) = -t \log t -(1-t)\log (1-t)$, 
we have $\binom{n}{m} \sim e^{n h(m/n)}$.
Hence, by using $t=x/n$, 
the binomial part of the pmf $p(x)$ in \eqref{eq:ac=0:p} can be approximated as
\begin{align}
\begin{split}
&\log \binom{n}{x} \binom{n}{m}^{-1} \binom{l}{x} \binom{m}{x}^{-1} \binom{n-l-x}{m-l} \\
&\sim  n \, h\Bigl(\frac{x}{n}\Bigr)-n \, h\Bigl(\frac{m}{n}\Bigr)
   +l \, h\Bigl(\frac{x}{l}\Bigr)-m \, h\Bigl(\frac{x}{m}\Bigr)
   +(n-l-x) \, h\Bigl(\frac{m-l}{n-l-x}\Bigr) = n f(t), 
\end{split}
\end{align}

Since $p(x)$ expresses a probability and 
the number of choices of $x$ is less than $n$, we have $f(t)\le 0$.
Also, since $f(t)$ is smooth, the equality $f(t)= 0$ implies $\frac{df}{dt}(t)=0$.
Then, using $\frac{dh}{dt}(t) = \log \frac{1-t}{t}$ and 
$\frac{d g}{dt}(t) = -\log (1-\frac{c}{t})$ for $g(t) \ceq t \, h(\frac{c}{t})$, we have 
\begin{align}
\begin{split}
\frac{df}{dt}(t)
=&\log \frac{1-t}{t}
 +\log \frac{1-\frac{t}{\alpha}}{\frac{t}{\alpha}}
 -\log \frac{1-\frac{t}{\alpha+\beta}}{\frac{t}{\alpha+\beta}}
 +\log\Bigl(1-\frac{\beta}{1-\alpha-t}\Bigr) \\
=&\log \frac{(1-t)(1-\alpha-\beta-t)(\alpha-t)}{t(1-\alpha-t)(\alpha+\beta-t)}.
\end{split}
\end{align}
Hence, $\frac{df}{dt}(t)=0$ if and only if
\begin{align}\label{eq:ac=0:AM}
 (1-t)(1-\alpha-\beta-t)(\alpha-t) = t(1-\alpha-t)(\alpha+\beta-t).
\end{align}
Using \eqref{eq:ac=0:munu}, we can find that 
the solutions of this cubic equation \eqref{eq:ac=0:AM} are
\begin{align}\label{eq:ac=0:sol}
 t = \tfrac{1}{2}, \mu, \nu.
\end{align}

Next we consider the cumulant generating function.
Using the variable $t=x/n$, we can compute it as 
\begin{align}
 e^{\phi_n(s)} = \bbE[e^{s X}] = \sum_{x \ge 0} e^{sx} p(x) \sim 
 \sum_{x \ge 0} \frac{1-2t+n^{-1}}{1-t+n^{-1}} \exp\Bigl(n \bigl(s t+f(t)\bigr)\Bigr).
\end{align}
Therefore, $\lim_{n\to \infty} \phi_n(s)/n$ is equal to 
$ u(s) = \max_{t}\bigl(s t+f(t)\bigr)$ defined in \eqref{eq:ac=0:u}.

Now, we turn to the evaluation of the tail probability.
We denote by $\Prb_n$ the probability for the distribution $P_{n,m,k,l}$.
The Markov inequality \cite[(2.168)]{Springer} implies that 
for any $R \in \bbR$ we have 
\begin{align}
 \Prb_n \bigl[ \tfrac{X}{n} \ge R \bigr]
 & \le \exp\Bigl(-n \cdot \max_{s \ge 0} (s R -\tfrac{\phi_n(s)}{n}) \Bigr),
 \label{eq:ac=0:MK3} \\
 \Prb_n \bigl[ \tfrac{X}{n} \le R \bigr] 
 & \le \exp\Bigl(-n \cdot \max_{s \le 0} (s R -\tfrac{\phi_n(s)}{n}) \Bigr). 
 \label{eq:ac=0:MK4}
\end{align} 
Since $u(s)$ is a smooth function, due to the G\"{a}rtner-Ellis theorem 
\cite[Theorem 2.3.6]{DZ}, \cite[Theorem 2.8]{Springer}, 
the opposite inequalities also hold with the limit.
Therefore, we obtain the relations \eqref{NBCV0} and \eqref{NBCV1}.

\section{Applications}\label{S-app}
It is known that the estimation of the unitary noiseless phase operation 
achieves the phase estimation via the Fourier analytical method.
When the phase operation on qubits has the standard phase damping noise,
the stochastic behavior of the estimation error 
is described by the distribution $P_{n,m,k,l}$ \cite[Appendix B]{HY}.
The paper \cite{HY} considers the case
when the probability of phase damping on one qubit is $p$
and we have $n$qubits.
When $p$ is fixed and $n$ increase, 
the error analysis on the phase operation
corresponds to 
the distribution $P_{n,m,k,l}$ with Type II limit.
When $p$ behaves as $c/n$ with a constant $c$ and $n$ increase, 
the error analysis on the phase operation
corresponds to 
the distribution $P_{n,m,k,l}$ with Type I limit.
Based on this correspondence,
the paper \cite{HY} employs the fact that
the limit pmf $q(x \midd \xi,k,l)$ is a convolution of two binomial distributions,
it is a polynomial of $\xi$.
This property was used to prove 
Lemma 1 of the paper \cite{HY}, which
takes the essential role 
to derive the Heisenberg scaling in the phase estimation
when the noise parameter $p$ behaves as $c/n$.
In this reason, 
our analysis on the distribution $P_{n,m,k,l}$
is essential for analysis of the phase estimation
with the standard phase damping noise.

As the second application, we point out the asymmetry with respect to permutation.
Asymmetry is one of resources in quantum information \cite{VAW,Marvian}.
Its amount is often measured by the mutual information.
Since the Dicke state $ \ket{\Xi_{N+M,M}}$ is invariant for permutation,
it does not have the asymmetry.
However, the state $ \ket{\Xi_{n,m|k,l}}$ is not invariant for permutation
so that it has a certain amount of invariant for permutation.
To calculate its amount, we need to discuss the distribution $P_{n,m,k,l}$.
In particular, the derivation of its asymptotic behavior needs
our asymptotic analysis on the distribution $P_{n,m,k,l}$.
In addition, to get its higher order asymptotic analysis, 
we need the exponential decreasing rate presented in Proposition \ref{prp:5.6.2}.
Since the analysis on the asymmetry is beyond the scope of this paper, 
this topic will be discussed in another paper \cite{HHY3}.

\section{Conclusion and discussion}\label{S5}
We have discussed the asymptotic behavior of the outcome of 
Schur-Weyl duality measurement under 
two kinds of settings.
The first setting addresses the case when the state is given as
the permutation mixture 
$\rho_{mix,n,l}$
of the state $\kb{1^{l} \, 0^{n-l}}$.
The second setting addresses the case when the state is given as
the tensor product
of the permutation mixture $\rho_{mix,k,l}$
and the Dicke state $ \ket{\Xi_{n-k,m-l}}$
and the distribution of the outcome is denoted by $P_{n,m,k,l}$.
We have derived various types of asymptotic distribution including 
a kind of central limit theorem when $n$ goes to infinity.

For our analysis on the distribution $P_{n,m,k,l}$,
we have employed the fact that the pmf $p(x \midd n,m,k,l)$ of 
the distribution $P_{n,m,k,l}$ is described by 
Hahn and Racah polynomials, which are 
$\hg{3}{2}$- and $\hg{4}{3}$-hypergeometric orthogonal polynomials.
Since the distribution $P_{n,m,k,l}$ has a highly complicated form,
this fact shows the usefulness of 
Hahn and Racah polynomials.
The relation between these polynomials and 
the distribution $P_{n,m,k,l}$ is based on the 
the Schur-Weyl duality, which is a key structure in quantum information.
Therefore, we can expect that 
these polynomials can be used in other topics in quantum information.
In particular,
for the above derivation of the central limit theorem under Type II limit,
we have employed the recurrence relation among three probabilities
$p(x),(x+1),p(x+2)$ that was obtained 
from the formulas for hypergeometric function in the paper \cite{HHY1}.
This fact shows that such a recurrence relation
is useful for the central limit theorem
Therefore, it is expected that a similar 
recurrence relation derives the central limit theorem.
Seeking such a possibility is another future topic.

Although Section \ref{S2-2} has presented
the asymptotic analysis under Type II limit under the assumption that
$\beta\delta>0$,
the analysis on this case already has been presented as
the first setting. 
Therefore, 
under Type II limit, we have obtained three types of distributions,
normal distributions, Rayleigh distributions, and geometric distributions
as the limit distribution
dependently of the parameters $\alpha,\beta, \delta,\gamma$.
That is, the limit distribution changes
discontinuously for these parameters.
This fact means that the convergence is not uniform
around the discontinuous points.
For example,
when $\xi(=\alpha+\beta)=\frac{1}{2}$ and $\beta\delta=0$,
the limit distribution is a Rayleigh distribution.
Hence, when $\beta=n^{-s}$ with parameter $0<s<1$
and $\xi=\frac{1}{2}$,
we can expect that the limiting distribution is an intermediate distribution
between normal distributions and Rayleigh distributions.
As an example of such distributions, we can list 
the Rayleigh-normal distribution, which was introduced in 
the reference \cite{KH}, and 
is a one-parameter family of distributions connecting normal and Rayleigh distributions.
Investigating the limit behavior of the distribution 
in this region is an interesting open problem.

Further, 
when $\xi\neq \frac{1}{2}$ and $\beta\delta=0$,
the limit distribution is a geometric distribution,
i.e., the limit distribution is a discrete distribution due to the following reason.
In this case,  since the variable $\frac{X-n \mu}{\sqrt{n}}$ converges to $0$ in probability,
we need to focus on the variable ${X-n \mu}$ that has a different scaling.
Hence, it is expected that the variable has 
a different scaling in intermediate cases, e.g., 
when 
$\beta\delta=0$ and 
$\xi= \frac{1}{2}+n^{-s}$ with parameter $0<s<1$
or when $\beta=n^{-s}$ with parameter $0<s<1$
and $\xi\neq \frac{1}{2}$.
It is another interesting open problem to clarify the limit behavior in these cases.


\section*{Acknowledgement}
M.H.\ was supported in part by the National
Natural Science Foundation of China under Grant 62171212.
A.H.\ was supported in part by JSPS KAKENHI Grant Number 19K03532. %
S.Y.\ was supported in part by JSPS KAKENHI Grant Number 19K03399.

\section*{References}


\end{document}